\documentclass[11pt]{article}
\usepackage[colorlinks]{hyperref}
\hypersetup{linkcolor=cyan,filecolor=cyan,citecolor=cyan,urlcolor=cyan}
\usepackage{xspace}
\usepackage{thm-restate,color,xcolor}
\usepackage{amsmath,amssymb,bbm,amsthm,algorithmicx,algorithm,algpseudocode}
\usepackage{fullpage}
\usepackage{boxedminipage}
\usepackage{amsmath}
\usepackage{mathtools}
\usepackage{framed}
\usepackage[framemethod=tikz]{mdframed}
\usepackage{titlesec}
\usepackage{lipsum}%
\usepackage{cleveref,aliascnt}
\usepackage{tikz}
\usetikzlibrary{shapes,arrows,positioning}

\usepackage{wrapfig}

\algblock{ParFor}{EndParFor}
\algnewcommand\algorithmicparfor{\textbf{parfor}}
\algnewcommand\algorithmicpardo{\textbf{do}}
\algnewcommand\algorithmicendparfor{\textbf{end\ parfor}}
\algrenewtext{ParFor}[1]{\algorithmicparfor\ #1\ \algorithmicpardo}
\algrenewtext{EndParFor}{\algorithmicendparfor}

\usetikzlibrary{shapes.geometric}
\usetikzlibrary{arrows}
\usetikzlibrary{arrows.meta}
\usetikzlibrary{patterns}
\usetikzlibrary{shapes.misc}

\newcommand{\dist}{\mbox{\rm dist}}
\newcommand{\rr}{\mathbb{R}}

\usepackage{placeins}

\newtheorem{theorem}{Theorem}[section]
\newtheorem{lemma}[theorem]{Lemma}
\newtheorem{meta-theorem}[theorem]{Meta-Theorem}

\newtheorem{corollary}[theorem]{Corollary}

\newtheorem{observation}[theorem]{Observation}
\newtheorem{definition}[theorem]{Definition}

\newcommand{\eps}{\varepsilon}
 
\newcommand{\dilation}{\mathsf{\ell}}
\newcommand{\congestion}{\mbox{\tt c}}

\newcommand{\poly}{\operatorname{\text{{\rm poly}}}}

\newcommand{\Sparsify}{\mathsf{Sparsify}\xspace}
\newcommand{\ExpDecomp}{\mathsf{ExpDecomp}\xspace}

\newcommand{\BoundedHopSparsify}{\mathsf{LCExpanderSparsify}\xspace}
\newcommand{\BoundedHopExpDecomp}{\mathsf{LCExpanderExpDecomp}\xspace}
\newcommand{\NeighborCover}{\mathsf{NeighborhoodCover}\xspace}
\newcommand{\MinDegree}{\mathsf{MinDegree}\xspace}
\newcommand{\FDSpanner}{\mathsf{FDSpanner}\xspace}
\newcommand{\GreedyFDSpanner}{\mathsf{GreedyFDSpanner}\xspace}
\newcommand{\FDCertificate}{\mathsf{FDCertificate}\xspace}

\newcommand{\Oish}{\widetilde{O}}

\def\Vol{\mbox{\tt Vol}}
\def\InitialCong{\mbox{\tt c}}

\newcommand{\bee}{\mathcal{B}}
\newcommand{\faultdeg}{f}
\newcommand{\Althofer}{Alth\"{o}fer}

\usepackage{booktabs}
\usepackage{multicol}
\usepackage{multirow}

\usepackage{todonotes}

\usepackage{geometry}
\begin{document}
\date{}

\title{Fault-Tolerant Spanners against Bounded-Degree Edge Failures: \\ Linearly More Faults, Almost For Free}

\author{Greg Bodwin\thanks{\texttt{bodwin@umich.edu}.  Supported by NSF:AF 2153680.}\\
University of Michigan
\and
Bernhard Haeupler\thanks{\texttt{bernhard.haeupler@inf.ethz.ch}. Supported by the European Research Council (ERC grant 949272).}\\
ETH Zurich \& CMU
\and
Merav Parter\thanks{\texttt{merav.parter@weizmann.ac.il}. Supported by the European Research Council (ERC grant 949083) and the Israeli Science Foundation (ISF), grant 2084/18.}\\
Weizmann Institute of Science}

\maketitle

\pagenumbering{gobble}

\begin{abstract}
We study a new and stronger notion of fault-tolerant graph structures whose size bounds depend on the degree of the failing edge set, rather than the total number of faults.
For a subset of faulty edges $F \subseteq G$, the \emph{faulty-degree} $\deg(F)$ is the largest number of faults in $F$ incident to any given vertex.
For example, a matching $F$ has $\deg(F)=1$ while $|F|$ might be as large as $n/2$.

We design new fault-tolerant structures with size comparable to previous constructions, but which tolerate every fault set of small faulty-degree $\deg(F)$, rather than only fault sets of small size $|F|$.
Thus, for example, our structures can tolerate a \emph{linear} number of edge faults with almost the same size bounds currently known for handling a \textit{single} edge failure, provided that the edge faults are arranged in a matching.
Our main results are:
\begin{itemize}
\item{\textbf{New FT-Certificates:}} For every $n$-vertex graph $G$ and degree threshold $\faultdeg$, one can compute a connectivity certificate $H \subseteq G$ with $|E(H)| = \widetilde{O}(\faultdeg n)$ edges that has the following guarantee: for any edge set $F$ with faulty-degree $\deg(F)\leq \faultdeg$ and every vertex pair $u,v$, it holds that $u$ and $v$ are connected in $H \setminus F$ iff they are connected in $G \setminus F$.
This bound on $|E(H)|$ is nearly tight.
Since our certificates handle some fault sets of size up to $|F|=O(\faultdeg n)$, prior work did not imply any nontrivial upper bound for this problem, even when $\faultdeg=1$.

\item{\textbf{New FT-Spanners:}} We show that every $n$-vertex graph $G$ admits a $(2k-1)$-spanner $H$ with $|E(H)| = O_k(\faultdeg^{1-1/k} n^{1+1/k})$ edges, which tolerates any fault set $F$ of faulty-degree at most $\faultdeg$.
This bound on $|E(H)|$ optimal up to its hidden dependence on $k$, and it is close to the bound of $O_k(|F|^{1/2} n^{1+1/k} + |F|n)$ that is known for the case where the \textit{total} number of faults is $|F|$ [Bodwin, Dinitz, Robelle SODA '22].
Our proof of this theorem is non-constructive, but by following a proof strategy of Dinitz and Robelle [PODC '20], we show that the runtime can be made polynomial by paying an additional $\text{polylog } n$ factor in spanner size.
%
%
%
\end{itemize}

Our techniques are based on an adaptation of the blocking set method used in previous work on fault tolerant spanners, as well as a new expander-based toolkit which translates the quality guarantees for expander routing into fault tolerance guarantees.

\end{abstract}

\newpage

\tableofcontents

\newpage

\pagenumbering{arabic}
\section{Introduction}

A basic problem in theoretical computer science and graph theory is to compress distance or connectivity information from an input graph $G$ into much smaller space than $G$ itself.
In its most basic form, this corresponds to computation of a sparse spanning subgraph (which captures connectivity information), or to computation of a sparse spanner (which captures approximate distances).
However, in some domains such as distributed computing and network design, it is common for a few nodes or edges of the input graph $G$ to temporarily \emph{fail} and be unusable while they await repair.
This motivates a more constrained problem, in which the goal is to find a sparse subgraph $H$ that still captures connectivity/distance information of $G$ even after any ```reasonable'' failure event that might occur.
This paper concerns the design of subgraphs of this kind.

\subsection{Connectivity Certificates}

Let us begin our discussion in the setting of connectivity.
A common formalization of connectivity preservation under failures is known as \emph{fault-tolerant connectivity certificates}:
\begin{definition} [EFT Connectivity Certificates \cite{NagamochiI92}]
Given a graph $G$, an edge-subgraph $H$ is called an $f$-edge fault tolerant (EFT) connectivity certificate if, for any nodes $u,v$ and for any set $F$ of $|F| \le f$ edges, there is a $u \leadsto v$ path in $G \setminus F$ iff there is a $u \leadsto v$ path in $H \setminus F$.
\end{definition}

This definition encodes a design choice: the idea of a ``reasonable'' failure event is formalized simply by choosing a size parameter $f$, and letting $F$ be any edge set of size $|F| \le f$.
We will refer to this as \textbf{global} fault tolerance, since it is counts faults equally no matter where they occur in the graph.
Global fault tolerance is the dominant model studied in prior work.
The tradeoff between certificate size and level of global fault tolerance for was resolved in a classic work of Nagamochi and Ibaraki \cite{NagamochiI92}:
\begin{theorem} [\cite{NagamochiI92}]\label{thm:NI}
For any positive integers $n, f$, every $n$-node graph $G$ has an $f$-EFT connectivity certificate $H$ on $|E(H)| = O(fn)$ edges, and this bound is best possible.
\end{theorem}

The main new ideas in Nagamochi and Ibaraki's work lie on the upper bound side.
The lower bound, that $|E(H)| = \Omega(fn)$ edges are sometimes necessary, is simple.
It is achieved by letting $G$ be any $(f+1)$-regular graph.
Consider an edge $(u, v) \in E(G)$, and let the failure set $F$ be the $f$ additional edges incident to $u$ besides $(u, v)$.
Then $(u, v)$ is the only edge still incident to $u$ in $G \setminus F$, which implies that we must keep the edge $(u, v)$ in the connectivity certificate $H$.
Since this holds for an arbitrary edge, it follows that the only $f$-EFT connectivity certificate of $G$ is $G$ itself, which contains $\Theta(fn)$ edges.

\begin{figure} [h]
\begin{center}
\begin{tikzpicture}
    \node [draw, circle] (u) at (0,0) {u};
    \node [draw, circle] (v) at (2,0) {v};
    \draw (u) -- (v);
    
    \node (a) at (-1.5,1) {};
    \node (b) at (-1.5,-1) {};
    \node (c) at (-2,0) {};
    \draw[dotted] (u) -- node [midway, red] {\Huge \bf $\times$} (a);
    \draw[dotted] (u) -- node [midway, red] {\Huge \bf $\times$} (b);
    \draw[dotted] (u) -- node [midway, red] {\Huge \bf $\times$} (c);

    \node [below left = 0.8cm and 0.4cm of u, red] {\bf $F$};
\end{tikzpicture}
\end{center}
\caption{By failing all edges incident to $u$ except for $(u, v)$, we see that any connectivity certificate must keep the edge $(u, v)$.}
\end{figure}
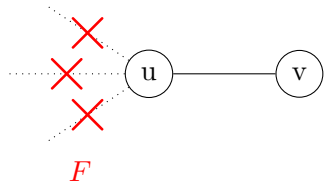

The $(f+1)$-regular graph lower bound is simple and optimal, but in a way it might also feel unsatisfying: is it really a ``reasonable'' failure event to see all $f$ edge failures incident to the same node $u$?
Although mass failures adjacent to $u$ are certainly possible -- failures can be correlated, and perhaps they all result from a common systemic problem with $u$ -- it is less reasonable to think that the rest of the graph would meanwhile remain completely failure-free.
Failures in different parts of the graph tend to be independent or positively correlated, and so when $u$ suffers failures, we may still expect to see \emph{at least} the usual number of failures occurring elsewhere in the graph at the same time.


This motivates a much stronger kind of connectivity certificate, which can tolerate edge failures that occur elsewhere in the graph alongside the $f$ failing edges incident to $u$.
We will consider a new \textbf{local} model of fault tolerance, which only limits failures-per-node rather than total number of failures:

\begin{definition} [Faulty-Degree Connectivity Certificates (New)]
Given a graph $G$, an edge-subgraph $H$ is called an $f$-faulty-degree (FD) connectivity certificate if, for any nodes $u,v$ and for any set of edges $F$ such that $\deg_F(x) \le f$ for all nodes $x$, there is a $u \leadsto v$ path in $G \setminus F$ iff there is a $u \leadsto v$ path in $H \setminus F$.
\end{definition}

An $f$-FD connectivity certificate is stronger than an $f$-EFT connectivity certificate, in the sense that every $f$-EFT connectivity certificate is also an $f$-FD connectivity certificate, but the converse is far from true.
Indeed, since an $f$-FD connectivity certificate must tolerate some failure sets of size up to $|F| = \Theta(fn)$, the previous upper bounds for $f$-EFT connectivity certificates do not imply \emph{any} nontrivial (subquadratic) upper bounds for $f$-FD connectivity certificates.
Despite the stronger model, the $(f+1)$-regular graph still only provides the same lower bound of $\Omega(fn)$ for $f$-FD connectivity certificates.
Our main result is that this is not an accident: there is actually a nearly-matching upper bound.

\begin{theorem} [First Main Result] \label{thm:introccmain}
For all positive integers $f, n$, every $n$-node graph has an $f$-FD connectivity certificate on $\Oish(fn)$ edges.
\end{theorem}

This theorem is constructive, i.e., we also provide an algorithm to compute the certificate in polynomial time.
The theorem proved in the paper gives a somewhat stronger guarantee; one can recover not just connectivity, but also shortest path distances up to a $\text{polylog } n$ factor.
Finally, we complement this theorem with a lower bound: although the hidden $\text{polylog } n$ factors in this upper bound might be improvable, they cannot be removed entirely.  This stands in contrast to the standard $f$-FT connectivity certificates of Theorem \ref{thm:NI}. Hence the cost of handling a faulty set with degree $f$ (with possibly $\Theta(fn)$ edges) is at least logarithmic and at most polylogarithmic, compared to handling a total of $f$ faults.

\begin{theorem}
There are $n$-node graphs for which any $f$-FD connectivity certificate has $\Omega(fn \cdot \frac{\log n}{\log f})$ edges.
\end{theorem}
Our lower bound construction is still relatively simple: when $f=1$, the certificate size lower bound is $\Omega(n \log n)$, and this is achieved by analyzing the hypercube graph.
For general $f$ the construction is essentially the extension of the hypercube graph to an alphabet of size $f$.

\subsection{Fault-Tolerant Spanners}

A \emph{spanner} is a subgraph that preserves approximate pairwise distances, rather than just connectivity.
Spanners and their variants have many applications in algorithms and network design; see survey \cite{ahmed2020graph}.
In this discussion, graphs can have arbitrary positive edge weights.

\begin{definition} [Spanners \cite{PelegS:89}]
For an input graph $G$, an edge-subgraph $H$ is called a \emph{$t$-spanner} of $G$ if it satisfies $\dist_H(u, v) \le t \cdot \dist_G(u, v)$ for all nodes $u,v$.
\end{definition}

The parameter $t$ of the spanner is often called its \emph{stretch}.
The existentially optimal size-stretch tradeoff was settled in a classic paper by \Althofer{}, Das, Dobkin, Joseph, and Soares:
\begin{theorem} [\cite{AlthoferDDJS:93}]
For all positive integers $n, k$, every $n$-node graph $G$ has a $(2k-1)$-spanner on $O(n^{1+1/k})$ edges.
\end{theorem}

This upper bound is unconditional, and there is a matching lower bound assuming the girth conjecture \cite{erdHos1964extremal}, i.e., this tradeoff is best possible up to the hidden constant.
Much like connectivity certificates, spanners are often applied in settings with edge faults, in which case a notion of fault tolerance is needed.
This was first considered in the special case of Euclidean input graphs by Levcopoulos, Narasimhan, and Smid \cite{levcopoulos1998efficient} and has since been studied intensively in this setting, and also in the more general setting of doubling metrics \cite{czumaj2004fault, lukovszki1999new, NS07, solomon2014hierarchical, chan2015sparse, chan2015new, le2023optimal}.
For general input graphs, the first results on fault-tolerant spanners were obtained by Chechik, Langberg, Peleg, and Roditty \cite{ChechikLPR:10}, in the following global model of fault-tolerance identical to the one used for connectivity certificates.

\begin{definition} [Fault-Tolerant Spanners \cite{ChechikLPR:10}]
Given a graph $G$, an edge-subgraph $H$ is an $f$-edge fault-tolerant ($f$-EFT) $t$-spanner if for any set $F$ of $|F| \le f$ edges and any nodes $u,v$, we have
$\dist_{H \setminus F}(u,v) \le t \cdot \dist_{G \setminus F}(u, v).$
\end{definition}

\begin{table}[t]
\begin{center}

    \begin{tabular}{llcl}
    \toprule
        \textbf{Stretch} & \textbf{Size} & \textbf{Polytime?} & \textbf{Citation} \\
    \midrule
        \multicolumn{4}{c}{\bf Global Fault Tolerance ($f$-EFT spanners)}\\
        \midrule
        $2k-1$ & $O \left(f \cdot n^{1+1/k} \right)$ & \checkmark{} & \cite{ChechikLPR:10} \\
        $2k-1$ & $O_k \left( f^{1 - 1/k} \cdot n^{1+1/k} \right) $  & & \cite{BDPW18} \\
        $2k-1$ & $O \left( f^{1 - 1/k} \cdot n^{1+1/k} \right)$ & & \cite{BP19} \\
        $2k-1$ & $O_k \left( f^{1 - 1/k} \cdot n^{1+1/k} \right)$ & \checkmark{} & \cite{DR20} \\
        $2k-1$ & $O \left( f^{1 - 1/k} \cdot n^{1+1/k} \right)$ & \checkmark{} & \cite{BodwinDR21}\\
        $2k-1$, $k$ even & $O_k \left( f^{1/2} \cdot n^{1+1/k} + f n \right)$ & \checkmark{} & \cite{BodwinDR22} \\
         $2k-1$, $k$ odd & $O_k \left( f^{1/2 - 1/(2k)} \cdot n^{1+1/k} + f n \right)$ & \checkmark{} & \cite{BodwinDR22} \\
         $3$ & $\Omega \left( f^{1/2} \cdot n^{3/2} \right)$ & lower bound & \cite{BDPW18} \\
         $2k-1, k>2$ & $\Omega \left( f^{1/2 - 1/(2k)} \cdot n^{1+1/k} + f n \right)$ & lower bound$^*$ & \cite{BDPW18}\\
         \midrule
        \multicolumn{4}{c}{\bf Local Fault Tolerance ($f$-FD spanners)}\\
        \midrule
        $\text{polylog } n$ & $\Oish(f \cdot n)$ & \checkmark{} & \textbf{this paper}\\
        $2k-1$ & $\Oish_k \left( f^{1-1/k} \cdot n^{1+1/k} \right)$ & \checkmark{} & \textbf{this paper}\\
        $2k-1$ & $O_k \left( f^{1-1/k} \cdot n^{1+1/k} \right)$ & & \textbf{this paper}\\
        $2k-1$ & $\Omega\left( f^{1-1/k} \cdot n^{1+1/k} \right)$ & lower bound$^*$  & \cite{BDPW18}\\
        $\infty$ (connectivity) & $\Omega\left(\frac{f}{\log f} \cdot n \log n \right)$ & lower bound & \textbf{this paper}\\

    \bottomrule
    \end{tabular}
    
    \caption{\label{tbl:priorftspan} Prior work on upper and lower bounds for edge fault-tolerant spanners, both in the global model ($f$-EFT spanners) and the local model ($f$-FD spanners).  The lower bound entries marked with an asterisk ($^*$) are conditional on the girth conjecture \cite{erdHos1964extremal}.  We also show constructions of FD-spanners via expander decomposition methods, not listed in this table.}
    \end{center}
\end{table}

We recap the considerable prior work on EFT spanners in Table \ref{tbl:priorftspan}.
For now, we note that the lower bound graphs from \cite{BDPW18}, which we describe in Section \ref{sec:exspanlb}, partially share the same locality property as the lower bounds for connectivity certificates.
In particular, we can consider the following model of local fault tolerance for spanners:
\begin{definition}[Faulty-Degree Spanners (New)]
Given a graph $G$, an edge-subgraph $H$ is called an $f$-faulty-degree (FD) $t$-spanner if, for any nodes $u, v$ and for any set of edges $F$ such that $\deg_F(x) \le f$ for all nodes $x$, we have $\dist_{H \setminus F}(u,v)\leq t \cdot \dist_{G \setminus F}(u,v)$.
\end{definition}

Once again, every $f$-EFT spanner is necessarily also an $f$-FD spanner, but the converse is far from true, as $f$-FD spanners must handle some fault sets of size up to $\Theta(fn)$.
The graph constructions from \cite{BDPW18} imply a lower bound of $\Omega(f^{1-1/k} n^{1+1/k})$ for $f$-FD spanners.
Our next main result is to show that this bound is essentially tight.

\begin{theorem} [Second Main Result] \label{thm:introexspanners}
For all positive integers $f, n, k$, every $n$-node graph has an $f$-FD $(2k-1)$ spanner on $f^{1-1/k} \cdot n^{1+1/k} \cdot O(k)^k$ edges.
Assuming the girth conjecture \cite{erdHos1964extremal}, this tradeoff is best possible when $k$ is a constant.
Additionally, one can compute the spanner in polynomial time by paying an additional $\text{polylog } n$ factor in the spanner size.
\end{theorem}

We obtain the result by analyzing a certain exponential-time greedy algorithm, analogous to the one used repeatedly in previous work on EFT spanners \cite{BDPW18, BodwinDR22}.
In order to speed it up to polynomial time, we follow an approach of Dinitz and Robelle \cite{DR20} to replace a key step in the greedy construction with an approximation algorithm.
Thus costs roughly the approximation factor in spanner size.
We overview the technical aspects of this approximation algorithm more in Section \ref{sec:lcc}.

Although our results focus on spanners and connectivity certificates, we wish to emphasize here that the local (FD) model of fault tolerance is sensible for many other problems, e.g., distance oracles, flows, etc. It will also be interseted to provide analog local FD results for vertex faults. 
The conceptual message behind Theorems \ref{thm:introccmain} and \ref{thm:introexspanners} is that, despite being considerably stronger, one can still prove meaningful upper bound results in the local model.
A natural open question is therefore to investigate whether upper bounds can be achieved when considering these other problems in the local fault tolerant model, which have previously been considered only in the global fault tolerant setting.

%

\subsection{Expanders and Rigid vs.\ Competitive Fault Tolerance}

An important theme in the literature on fault-tolerant graph structures is the distinction between \emph{rigid} and \emph{competitive} fault tolerance \cite{ChechikLPR:10}.
All of the previous discussion concerns the competitive notion of fault tolerance, because we compare distances or connectivity between the post-failure subgraph $H \setminus F$ and the \emph{post-failure} subgraph $G \setminus F$.
The rigid notion of fault tolerance is defined similarly, but we would compare distances or connectivity between $H \setminus F$ and the \emph{pre-failure} subgraph $G$.

Rigid fault tolerance is a stronger notion than competitive fault tolerance, and hence it is preferable when it can be achieved.
However, rigid fault tolerance is not available for all graphs.
For example, when $G$ is a tree, a single edge failure can disconnect the graph.
Thus there is no $1$-rigid-fault-tolerant subgraph of $G$, not even $G$ itself.
Despite this, rigid fault tolerance is available for certain restricted graph classes, e.g., Euclidean \cite{levcopoulos1998efficient}.
An additional contribution of this work is a refined understanding of when rigid fault tolerance is available.

\paragraph{Rigid Fault Tolerance for Connectivity.}

A key technical ingredient in our construction of connectivity certificates is a new structural lemma, showing that every \emph{expander} of high enough minimum degree admits a notion of rigid fault tolerance for connectivity (and, in fact, also for expansion).

\begin{theorem} [Expanders are Rigid to Bounded-Degree Faults] \label{thm:expander-d-deg}
Let $G$ be an $n$-vertex $\phi$-expander with minimum degree $\faultdeg' \geq 2\faultdeg/\phi$.
Then for every $\faultdeg$-degree faulty-set $F$, the graph $G \setminus F$ remains connected.
(In fact, $G \setminus F$ is itself a $\phi/2$-expander.)
\end{theorem}

For general graphs $G$, this implies a strategy to design FD connectivity certificates.
Suppose that we compute an expander decomposition of $G$ (or rather, a slight variant of expander decomposition that also enforces a high enough minimum degree of each expander).
We then sparsify each expander in the decomposition, while maintaining its expansion properties and its minimum degree.
Theorem \ref{thm:expander-d-deg} promises that every edge with both endpoints in the same expander is \emph{automatically} protected from faults: these edge failures simply cannot disconnect the graph, so they can be ignored.
Indeed, the only bad edges are the ones whose endpoints lie in two different expanders in the decomposition.
These bad edges can be handled recursively, until very few of them remain, and the remaining edges can simply be added to the certificate.
We overview this strategy in more detail in Section \ref{sec:techocertificates}.

\paragraph{Rigid Fault Tolerance for Distances.}

Theorem \ref{thm:expander-d-deg} actually provides a notion of rigid fault tolerance for \emph{distances}, rather than just connectivity.
Since the post-failure graph $G \setminus F$ remains an expander, its distances can be changed only by a factor of $O(\log n)$ by the fault set $F$ -- in other words, $G \setminus F$ is an $O(\log n)$ spanner of $G$.
Plugging this fact into the recursive expander decomposition framework outlined above, this ultimately implies that the FD connectivity certificate is in fact an FD spanner with stretch $\text{polylog } n$.

It is natural to ask whether this $O(\log n)$ factor in rigid fault tolerance for distance can be improved.
We show that it can, by instead considering the recently-introduced class of length-constrained (LC) expanders \cite{HaeuplerR022}.
We will review these objects formally in Section \ref{sec:techospanners}, but briefly, an $(h,s)$-length expander is a graph where any $h$-length unit demand can be routed by a multi-commodity flow with congestion $\widetilde{O}(1/\phi)$ and over paths of length at most $hs$. 
The following theorem shows that they provide an improved notion of rigid fault tolerance for distances:

\begin{theorem}[Length-Constrained Expanders are Rigid to Bounded-Degree Faults]\label{thm:bounded-hop-expander-robust}
Let $G$ be an $n$-vertex $(h,s)$-length $\phi$-expander with minimum degree $\widetilde{\Omega}(\faultdeg \cdot n^{\epsilon}/\phi)$.
Then for every $\faultdeg$-degree faulty-set $F$, the graph $G \setminus F$ is a $(hs)^{O(1/\eps)}$-spanner of $G$.
\end{theorem}

The analogous construction now works for FD spanners: one can take a LC expander decomposition (again enforcing high minimum degree), and then sparsify each expander while maintaining its expansion properties and its minimum degree. The precise arguments are more delicate in this setting and in particular, our sparsification does not provide an LC-expander but rather a sparse subgraph that has similar routing properties as those provided by LC expanders. 
But the effect is again that edge faults whose endpoints lie in the same expander can essentially be ignored, since Theorem \ref{thm:bounded-hop-expander-robust} implies that these faults simply cannot change the pairwise distances by too much.
The only bad edges are the ones whose endpoints lie in different expanders, which are again handled recursively, until only a few bad edges remain which can be taken into the spanner.
Analyzing this construction leads to the following result:

\begin{theorem} [Second Main Result, Part 2] \label{thm:introexpanderspanners}
For all positive integers $f, n, k$, every $n$-node unweighted graph has an $f$-FD $k^{O(k)}$-spanner $H \subseteq G$ with $|E(H)|=\widetilde{O}(f \cdot n^{1+1/k})$ edges.
\end{theorem}

Although the size-stretch tradeoff of the FD-spanner in Theorem \ref{thm:introexpanderspanners} is worse than the one in Theorem \ref{thm:introexspanners} (which is achieved using an unrelated non-expander-based toolkit), we include it because we think that the new connections between expanders, expander decomposition, and rigid fault tolerance illustrated by this theorem are worthwhile to explore, and could have applications in followup work, especially where computational aspects of the construction are emphasized.
Additionally, the non-tightness of the size-stretch tradeoff in this theorem is mostly due to a lack of understanding of the distance rigidity properties of expanders, rather than a weakness in the expander-based construction framework itself.
The following result states that the size-stretch tradeoff in this construction could improve from exponential to linear -- hence tight, up to a constant factor in the stretch and a sublinear factor in the faulty-degree -- if the distance rigidity bounds from Theorem \ref{thm:bounded-hop-expander-robust} can be improved.

\begin{theorem} [Informal]
Suppose that the distance rigidity bound in Theorem \ref{thm:bounded-hop-expander-robust} can be improved from $O(hs)^{O(1/\eps)}$ to $\ell$, for some parameter $\ell$.
Then our construction yields an $f$-FD $O(\ell \cdot \poly(s))$-spanner $H \subseteq G$ with $|E(H)|=\widetilde{O}(f \cdot n^{1+O(1/s)})$ edges.
\end{theorem}

We think it is an interesting open problem to improve the distance rigidity bound $\ell$ in this theorem, both in its own right and to improve this particular application to FD spanners.
We refer to Section \ref{sec:techospanners} for more technical details on the construction and this possible extension.

\smallskip
\noindent \textbf{Related Work on Spanners and Expanders.} In a very recent work, Haeupler, Hershkowitz and Tan \cite{GreedyMatchingHHT23} provided a new implementation of the well-known greedy (non-FT) spanner algorithm of \cite{AlthoferDDJS:93}. In particular, they show that one can add at a time a collection of matching edges $M$ provided that each individual edge in $M$ does not close a short cycle in the current spanner. On a high-level, their size analysis approach is based on exploiting the routing properties of the recently introduced notion of length-constrained expanders \cite{HaeuplerR022}. While their paper does not address fault-tolerance aspects, their use of expanders and more specifically, expander routing, in the context of (standard) spanners highly inspired our expander-based constructions of FD spanners.

\section{Technical Overview} \label{sec:tech-overview}

\subsection{Expander-Based Constructions}\label{sec:expander-based}
Our expander-based constructions use a two-step approach.
We first show that, when the input graph is an expander, one can provide rigid fault-tolerant guarantees for connectivity and distances.
By employing the expander-decomposition technique, we then translate the rigid fault-tolerant guarantees for expanders into competitive fault-tolerant guarantees for general input graphs.
To the best of our knowledge, these tools have not been employed before in the context of fault-tolerant graph sparsification. 

On a conceptual level, our expander-based sparsification approach allows us to pinpoint a critical set of edges in the input graph $G$ that prevent one from providing rigid fault-tolerant guarantees for $G$. Our proof shows that the \emph{only} problematic edges for rigid fault-tolerance are the inter-expander edges, that connect different expander subgraphs in the output of the expander decomposition procedure. Those edges can be handled recursively until their number becomes bounded by $\widetilde{O}(\faultdeg \cdot n)$. This set of small number of edges are the problematic edges for obtaining rigid fault-tolerance, and by adding them to the output subgraph, we provide competitive fault-tolerance for $G$.

\subsubsection{Sparse Certificates \label{sec:techocertificates}}

We first discuss our use of expanders for sparse connectivity certificates.
Our starting observation is that any $\phi$-expander with minimum degree $\Theta(f)$ can tolerate the removal of any faulty-set $F$ with degree $f$, and maintain its expansion properties, up to a small loss in the conductance parameter.  As  formulated in Theorem \ref{thm:expander-d-deg}.

To illustrate our ideas, assume for simplicity that the minimum degree of the input graph $G$ is $\widetilde{\Omega}(\faultdeg)$. Our algorithm first employs an expander-decomposition procedure on $G$ which partitions $G$ into vertex-disjoint $\phi$-expanders $G_1,\ldots, G_k$ for $\phi=1/\log n$, and an additional set $H_0 \subseteq G$ of at most $m/2$ inter-expander edges. In the second step, the algorithm sparsifies each expander $G_i$ into a $\Theta(\phi)$-expander $H_i \subseteq G_i$ with only $\widetilde{O}(\faultdeg \cdot |V(G_i)|)$ edges. The final subgraph is then given by $H=\bigcup_{i=0}^k H_i$. 
Formally, our expander sparsification lemma can be stated as follows:

\begin{lemma}[Expander Sparsification]\label{lem:exp-sparsify}
Given a $\phi$-expander $n$-vertex $G$ with minimum degree $\widetilde{\Omega}(\faultdeg/\phi)$, one can compute a subgraph $H \subseteq G$ such that (i) $|H|=\widetilde{O}(\faultdeg \cdot n/\phi^2)$, (ii) $H$ is an $n$-vertex $\widetilde{\Omega}(\phi^2)$-expander (and $V(H)=V(G)$), and (iii) the minimum-degree of $H$ is $\Omega(\faultdeg)$. 
\end{lemma}

At a high level, the sparsification is based on the useful notion of graph embedding (see e.g., Section 3.16 in \cite{LeightonR99}). An embedding $\sigma$ of a graph $\widehat{G}=(V, \widehat{E})$ into a graph $G=(V,E)$ is a function that maps each edge $e'$ in $\widehat{G}$ to a path $\sigma(e')$ in $G$. We say that congestion of the embedding is at most $\congestion$ if each $G$-edge $e$ appears on at most $\congestion$ paths in $\mathcal{P}=\{\sigma(e') ~\mid~ e' \in \widehat{E}\}$. The dilation of the embedding is the length of the longest path in $\mathcal{P}$. 

We then show that one can embed an $\faultdeg$-regular $\phi$-expander $\widehat{G}=(V,\widehat{E})$ into any given $\phi$-expander $G=(V,E)$ with minimum degree $\faultdeg$. The congestion and the dilation of the embedding are bounded by $\widetilde{O}(1/\phi)$. The output subgraph $H$ is taken as the union of the paths that embed $\widehat{G}$ into $G$.

Finally, we explain the intuition for the correctness of our algorithm. Fix a $\faultdeg$-degree subset $F$ and an edge $(u,v) \notin F$. We show that $\dist_{H \setminus F}(u,v)=\widetilde{O}(1)$. The interesting case is when $(u,v) \notin H$, and hence $(u,v)$ is in some $\phi=1/\log n$ expander $G_i$, obtained by the expander-decomposition. Since $H_i=(V(G_i), E_i)$ is a $\widetilde{\Omega}(\phi^2)$-expander with minimum degree $\Omega(\faultdeg)$, by Theorem \ref{thm:expander-d-deg}, $H_i \setminus F$ is also a $\widetilde{\Omega}(\phi^2)$-expander, hence of diameter $\widetilde{O}(1)$, which provides the desired stretch guarantees. 


\subsubsection{Sparse Spanners \label{sec:techospanners}}

\paragraph{Warmup: FD $3$-Spanners} As a warm-up, we show a simple construction of $\Delta$-FD $3$-spanners that can be obtained by a mild adaptation of the classic Baswana-Sen spanner algorithm \cite{BaswanaS:07}.
The following proof assumes familiarity with the Baswana-Sen algorithm, and proceeds quickly over some details used in this algorithm.
We have:
\begin{lemma}\label{lem:threespanner}
For any graph $G=(V,E)$, one can compute a $\faultdeg$-FD $3$-spanner $H \subseteq G$ with $|E(H)|=\widetilde{O}(\faultdeg n^{3/2})$.
\end{lemma}
\begin{proof}
Initially set $H=\emptyset$ and add to $H$ all edges incident to vertices with degree at most $\faultdeg \sqrt{n}$.
We will call these \emph{low-degree vertices}, and the others are \emph{high-degree vertices}.
Let $S$ be a random sample of $O(\sqrt{n}\log n)$ vertices, which we call \emph{cluster centers}.
With high probability we sample at least $\faultdeg+1$ cluster centers adjacent to each high-degree vertex.
Connect each high-degree vertex to an arbitrary set of $(\faultdeg+1)$ of its neighboring sampled centers in $S$.
Finally, for each high-degree vertex $v$ and each cluster center $c \in S$, we choose $3 \faultdeg$ arbitrary vertices that are neighbors of both $v$ and $c$, and we add the edge connecting $v$ to these vertices.
(If there are fewer than $3 \faultdeg$ vertices that are neighbors of both $v$ and $c$, add them all.)
Overall, this procedure adds $\widetilde{O}(\faultdeg n^{3/2})$ edges, as desired.

To prove correctness, it suffices to consider an edge $(u,v) \notin F$ where both $u,v$ are high-degree, where $F$ is an arbitrary $\faultdeg$-degree set of edge faults.
Since $u$ is adjacent to at least $\faultdeg+1$ cluster centers, there exists a center $c(u) \in S$ such that the edge $(c(u),u)$ survives in $G \setminus F$.
If $(u,v) \notin H$, then we have added $\geq 3\faultdeg$ edges to $H$ connecting $v$ to neighbors in the cluster of $c(u)$.
As $c(u),v$ are incident to at most $2\faultdeg$ edges in $F$, there exists at least one such vertex $x$ such that $(x,c(u)),(x,v) \notin F$.
Thus the path $(u,c(u), x, v)$ survives in $G \setminus F$, and so we have $\dist_{H \setminus F}(u,v)\leq 3$. 
\end{proof}

\paragraph{FD-Spanners with Constant Stretch.} Unfortunately, directly extending the (FT) Baswana-Sen approach to the FD-setting, as in \cite{parter2022nearly}, seems to be quite nontrivial. The reason is that FT Baswana-Sen is based on $\faultdeg$ edge-disjoint clustering which can handle a total of $\faultdeg$ faults but cannot handle a $\faultdeg$-degree subsets $F$ of $\faultdeg \cdot n$ failed edges. 

Our approach is based on replacing the low-depth edge-disjoint trees that form the clustering in the FT Baswana-Sen algorithm by length-constrained expanders \cite{HaeuplerR022} with a sufficiently large minimum degree.  Our approach of connecting spanners to length-constrained expanders is inspired by the very recent work of \cite{GreedyMatchingHHT23}. They use expanders to provide an alternative, matching-based, implementation of the well-known greedy algorithm for (non fault-tolerant) spanners.


\paragraph{Step (I): From Expander Routing to Rigid Fault-Tolerant Properties.} Our approach is based on translating the routing quality
of 
LC-expanders into rigid FT properties against bounded-degree faults. Our starting observation is the following.
Consider an $(h,s)$-length $\phi$-expander with minimum degree $\widetilde{\Omega}(\faultdeg \cdot n^{\epsilon} /\phi)$. We use the fact that length-constrained expanders are good \textit{routers} for $h$-hop pairs. Specifically, given an $\faultdeg$-degree set $F$, one can route $d=\widetilde{\Omega}(n^{\epsilon}/\phi)$ units of flow over each edge $(u,v) \in F$ along paths of length $O(h \cdot s)$ and with congestion\footnote{The congestion of a path collection $\mathcal{P}$ is at most $\congestion$ if each edge $e$ appears in at most $\congestion$ paths in $\mathcal{P}$.} of $\congestion=O(\log n/\phi)$.

We use these dilation and congestion bounds to deduce that $G\setminus F$ is an $(hs)^{O(1/\epsilon)}$ spanner of $G$, as follows: 
The solution to the routing instance consists of $d$ many $u$-$v$ paths of length at most $hs$ for every $(u,v)\in F$. Since these paths might intersect $F$, we cannot provide an immediate guarantee on the $u$-$v$ distance in $G \setminus F$ for any $(u,v)\in F$. However, by using an averaging argument, one can claim that for at least $(1-\congestion/d)$ fraction of the edges in $F$, their distance (between their endpoints) in $G \setminus F$ is at most $hs$. We then apply an iterative re-routing procedure that allows us to translate the original routing paths into those that provide at least one fault-free path (among the $d$ many paths) for any $(u,v)\in F$. This re-routing procedure has $\log_{\congestion/d} n=\Theta(1/\epsilon)$ iterations. In each iteration $i$, we can provide fault-free paths of length $(hs)^i$ for $(1-\congestion/d)$ fraction of the remaining pairs in $F$. This provides the intuition for Thm. \ref{thm:bounded-hop-expander-robust}.


\paragraph{Step (II): From Rigid-FT in Expanders into Competitive-FT in General Graphs.} 
For simplicity, we present the high-level ideas when using the existential bounds of LC expander decomposition. Our actual algorithms can be implemented in polynomial time, and thus use expander decomposition with somewhat weaker bounds.

To translate the rigid-FT properties of Thm. \ref{thm:bounded-hop-expander-robust} into a construction of FD spanners, we again employ the following two steps: expander decomposition and expander sparsification. Due to the nature of length-restricted expanders, the sparsification arguments are slightly more delicate. The output of the LC expander-decomposition provides us with a $(h,s)$ LC-expander\footnote{Unlike standard expanders, LC expanders are not required to be connected, hence the output of the expander decomposition is a possibly disconnected LC expander.} $G'$ with minimum degree $\widetilde{\Omega}(\faultdeg \cdot n^{O(1/s)})$ and a set of cut edges $C=G \setminus G'$ with $\widetilde{O}(\faultdeg \cdot n^{1+1/s})$ edges. 

Our goal is then to sparsify $G'$ into a subgraph $H' \subseteq G$ with a total number of $\widetilde{O}(\faultdeg \cdot n^{1+O(1/s)})$ edges, such that $H'$ has approximately the same routing properties as $G'$ in a way that allows us to apply the robustness lemma of Thm. \ref{thm:bounded-hop-expander-robust} on $H'$. The sparsification is based on 
embedding a collection of virtual expanders $\widehat{G}_1,\ldots, \widehat{G}_\ell$ into $G'$. The output $H'$ is taken as the union of the embedding paths in $G'$. Each virtual graph $\widehat{G}_i$ is defined based on a subset $V_i \subseteq V(G)$ of pairwise distances at most $h$ in $G'$. These subsets are obtained by using the well-known tool of neighborhood-cover of \cite{AwerbuchBCP98}. 



We believe that our approach for translating rigid-FT properties in expander into competitive-FT in general graphs might be also useful for other graph properties (beyond connectivity and distances) and under a wide collection of fault models. The high-level recipe is given some desired graph property $\Pi$ and a faulty model $M$ and has the following flow: (i) showing that any expander $G'$ approximates $\Pi$ under any failing event $F$ (where $F$ is determined by the faulty model $M$). (ii) showing that one can sparsify the expander into $H' \subseteq G'$ such that $H'$ is an expander that satisfies (i), and (iii) apply expander-decomposition (possibly recursively) accompanied with  expander sparsification. At the point where the number of inter-expander edges is sufficiently small, add them to the desired output subgraph $H$. We also note that while we focus here on connectivity and distances, as a byproduct of our approach, we also get that our output subgraphs (e.g., FD certificates) also approximate the dilation and congestion bounds of routing instances in $G$. 

\subsection{Blocking Set-Based Constructions}\label{sec:exist-stretch}

\subsubsection{The FD Greedy Algorithm and Analysis}

The near-optimal size/stretch tradeoff for FD spanners in Theorem \ref{thm:introexspanners} is proved by analyzing the natural FD adaptation of the greedy spanner algorithm \cite{AlthoferDDJS:93}, also analogous to the EFT greedy algorithm studied in prior work \cite{BDPW18, BodwinDR22}.
We state this algorithm formally in Section \ref{sec:exspannersetup}, but to quickly overview: the spanner $H$ is initially empty, we consider the edges $(u, v)$ of the input graph one at a time in order of increasing weight, and we add each edge $(u, v)$ to the spanner iff there exists a possible fault set $F$ (of max degree $\le f$) under which $\dist_{H \setminus F}(u, v) > (2k-1)\cdot w(u, v)$.

The proof of correctness, i.e., that the algorithm returns a correct $f$-FD $(2k-1)$-spanner of the input graph $G$, is standard.
The challenging part of the proof is to show that the output spanner $H$ does not contain too many edges.

\paragraph{The Blocking Set Method.}

The analysis of the non-fault-tolerant greedy spanner algorithm \cite{AlthoferDDJS:93} uses crucially that the output spanner has \emph{high girth}; that is, one can prove that all cycles in $H$ have $> 2k$ edges.
It is not still true that the output spanner $H$ from the FD greedy spanner algorithm must have high girth.
Still: since the FD greedy algorithm is similar in spirit to the non-fault-tolerant greedy algorithm, it is intuitive to think that its output spanner might be ``close to high girth'' in some structural sense.

A blocking set formalizes the idea of a graph being \emph{structurally close to high girth}, by asserting that its short cycles admit a particularly small or simple kind of hitting set.
The following definition of blocking set is the relevant one for this work:
\begin{definition} [FD Blocking Sets]
Let $H = (V, E)$ be a graph equipped with a total ordering of its edge set $E$.
A $\faultdeg$-fault-degree (FD) $k$-blocking set $\bee$ for $H$ is a set of pairs of the form $(e, F_e)$, such that:
\begin{itemize}
\item Each edge $e$ is the first edge of exactly one pair $(e, F_e)$, each $F_e$ is a set of edges from $E(H)$ of degree $\deg(F_e) \le \faultdeg$, and each edge in $F_e$ strictly precedes $e$ in the edge-ordering of $H$.
\item For each cycle $C$ in $H$ on $|C| \le k$ edges, letting $e$ be the latest edge in $C$ in the edge-ordering of $H$, we have that $F_e \cap C$ is nonempty.
\end{itemize}
\end{definition}

It is fairly easy to prove that the output spanner $H$ of the FD-greedy algorithm, with its edges ordered by their arrival in the algorithm, has a $\faultdeg$-FD blocking set.
The focus of the proof then shifts from bounding the number of edges in the particular output spanner from the FD-greedy algorithm to bounding the number of edges in \emph{any} graph that admits a $\faultdeg$-FD blocking set.

\paragraph{Path-Counting Methods.}

In order to explain how we limit the size of a graph with a $\faultdeg$-FD blocking set, it may be helpful to first recall the proof of the \emph{Moore bounds}, which are used to bound the maximum possible number of edges in an $n$-vertex graph $H$ of girth $>2k$.
The Moore bounds are proved using a counting argument over the \emph{edge-simple $k$-paths} of $H$.
They include two steps:
\begin{itemize}
\item A \emph{dispersion lemma}, which shows that no two edge-simple $k$-paths in $H$ can share endpoints (or else they would form a short cycle), and
\item A \emph{counting lemma} which shows a lower bound on the number of edge-simple $k$-paths in $H$, where this lower bound is increasing with the number of edges $|E(H)|$.
\end{itemize}
By comparing the upper and lower bounds on the number of edge-simple $k$-paths that respectively arise from the dispersion and counting lemmas, and rearranging terms, one gets an upper bound on $|E(H)|$.
The interested reader can refer to \cite{bodwin2023alternate}, Section 2, for a recap of this proof of the Moore bounds in full detail.

Since a graph with a blocking set is conceptually interpreted as a graph that is ``close to high girth,'' it stands to reason that our upper bound on the size of a graph with a blocking set should follow the same basic strategy as the Moore bounds.
Indeed, our strategy is to prove analogous dispersion and counting lemmas, and our final bound on $|E(H)|$ follows by comparing these bounds to each other and rearranging terms.
However, in our arguments we do not consider \emph{any} edge-simple $k$-paths in these lemmas: a graph with a $\faultdeg$-FD-blocking set can unfortunately have an unbounded number of edge-simple $k$-paths that share endpoints, and so the dispersion lemma would be impossible.
Instead, we restrict our attention to a subset of edge-simple $k$-paths satisfying a very particular set of properties, which are restrictive enough to enable a version of the dispersion lemma but not so restrictive as to break the counting lemma.
We call these MUCk paths.
This is an acronym, and the technical properties encoded by this acronym and their role in the proof are discussed more in Section \ref{sec:muck}).

\subsubsection{Speeding Up the Greedy Algorithm \label{sec:lcc}}

The main downside of the FD greedy algorithm is that, in a naive implementation, it takes exponential time.
The part that takes exponential time is that in each round, we need to check whether or not there exists a bounded-degree fault set $F$ for which $\dist_{H \setminus F}(u, v) \le (2k-1) \cdot w(u, v)$.
There are $\exp(n)$ many fault sets to consider, and it is not clear how to avoid checking each of these potential fault sets by hand.
In fact, there is some evidence that no substantial improvement on this brute force approach will be possible.
For the EFT greedy algorithm (which is identical except that it searches only over fault sets of size $|F| \le f$), the search for a valid fault set encodes an NP-hard problem called \textsc{Length-Bounded Cut (LBC)} \cite{BEHKSS06}:
\begin{mdframed}[backgroundcolor=gray!20]
\noindent \textsc{Length-Bounded Cut (LBC)}:
Given a graph $G = (V, E)$, vertices $s, t$, and an integer $k$, find the least integer $f$ for which there exists an edge set $F$ of size $|F| \le f$ with $\dist_{G \setminus F}(s, t) > k$.
\end{mdframed}

A nice paper by Dinitz and Robelle \cite{DR20} provides a method to escape this NP-hardness barrier for the EFT greedy algorithm.
Dinitz and Robelle essentially show an $O(k)$-approximation algorithm for \textsc{LBC} in unweighted graphs.
Using this approximation algorithm for the edge-test step of the greedy algorithm ultimately costs a factor of $O(k)$ in the size of the blocking set for the output spanner, which translates to an $O(k)$ factor in spanner size as the price to pay for polynomial runtime.
The other important observation in \cite{DR20} is that it suffices to solve \textsc{LBC} in an unweighted graph, even when the goal is to build a spanner of a weighted input graph.

It is natural to attempt the same method to improve the runtime of the FD greedy algorithm.
The same high-level proof strategy works, but the catch is that we need to solve the following variant:
\begin{mdframed}[backgroundcolor=gray!20]
\noindent \textsc{Min Max LBC}:
Given a graph $G = (V, E)$, vertices $s, t$, and an integer $k$, find the least integer $\faultdeg$ for which there exists an edge set $F$ of degree $\deg(F) \le \faultdeg$ with $\dist_{G \setminus F}(s, t) > k$.
\end{mdframed}

As it happens, so-called ``min max'' cut problems of this type have been studied before through an independent motivating framework, in the research area of \emph{correlation clustering}.
Correlation clustering was introduced by Bansal, Blum, and Chawla \cite{bansal2004correlation}, and it concerns the following model.
We are given a (possibly weighted) graph $G$, with some edges labeled $(+)$ and others labeled $(-)$.
We cluster the vertices of $G$, and a $(+)$ edge is considered to be satisfied if its endpoints lie in the same cluster, while a $(-)$ edge is considered to be satisfied if its edges lie in different clusters.
The general goal is to minimize the total weight of unsatisfied edges, although several different objective functions have been studied \cite{wirth2010correlation}.
Correlation clustering simultaneously generalizes several fundamental graph problems, including \textsc{Min s-t Cut}.

More recently, researchers have considered \emph{local} objective functions, which encode goals such as minimizing the maximum disagreement weight on any given vertex.
This framework was first studied by Charikar, Gupta, and Schwarz \cite{charikar2017local}, who obtained a $O(\sqrt n)$ approximation algorithm for the problem of local correlation clustering.
They also pointed out that, in the same way that correlation clustering captures \textsc{Max s-t Cut}, their algorithm implies a $O(\sqrt n)$ approximation algorithm for \textsc{Min Max s-t Cut} in undirected graphs.
This algorithm has been substantially extended in followup work \cite{kalhan2019correlation, jafarov2021local}.

A central message of the work of Charikar et al.~\cite{charikar2017local} is that graph problems of these type are fundamental, and warrant further study.
We fully agree with this message, and indeed, this paper provides a concrete application for this direction of study.
Our contribution is the following approximation algorithm, which when plugged into the framework of Dinitz and Robelle \cite{DR20}, implies our near-optimal FD spanners in polynomial time from Theorem \ref{thm:introexspanners}.

\begin{theorem}
The problem \textsc{Min Max LBC} on an $n$-vertex input graph with length parameter $k$ has a $O(k \log n)$ approximation algorithm.
\end{theorem}





\section{Preliminaries}

%

Throughout, we denote by $n$ the number of vertices in the graph input graph. As usual the notions $\widetilde{O}(\cdot)$ and $\widetilde{\Omega}(\cdot)$ hide poly-logarithmic factors in $n$.
\paragraph{Vertex Weightings.} A vertex weighting $W: V \mapsto \mathbb{Z}_{\geq 0}$ assigns each vertex $v$ a positive integer weight $W(v)$. For a multi-set of undirected edges $E'$ (potentially with (multiple) self-loops a vertices) the vertex weighting $W_{E'}$ assigns each vertex $v$ the number of edges adjacent to $v$ in $E'$. In this way any (multi-)graph $G=(V,E)$ comes with the so-called degree weighting $deg_G = W_G = W_E$ which assigns each vertex $v$ the weight $W(v) = deg_G(v)$.

Our algorithms use the notion of vertex weighting only in the context of length-constrained expanders, nevertheless to make the framework general we provide all definitions based on a given vertex weighting $W$. If no $W$ is provided then $W(u)=\deg_G(u)$ for every $u$.
\paragraph{Conductance, Volume, and Expanders.} Consider a graph $G=(V,E)$. For a vertex subset $S \subseteq V$ and vertex weighting $W$, let $\Vol_{W}(S)=\sum_{v \in S} W(v)$. We sometimes write $\Vol_{G}$ as a shorthand for $\Vol_{deg_G}$. Let $\partial_G(S)=E(S,V\setminus S)$ be the set of edges in $G$ with one endpoint in $S$ and the other endpoint in $V \setminus S$. The \emph{conductance} of a cut $S \neq \emptyset$ with respect to vertex weighting $W$ is defined by
$$\Phi_G(S,W)=\frac{|\partial_G(S)|}{\min\{\Vol_{W}(S), \Vol_{W}(V \setminus S)\}}~.$$

For $S=\emptyset$, let $\Phi(S,W)=0$. The \emph{conductance} of a graph $G$ is given by 
$\Phi(G,W)=\min_{S \subset V, S \neq \emptyset} \Phi_{G,W}(S)~.$
That is, $\Phi(G,W)$ is the minimum value of $\Phi_{G,W}(S)$ over all non-trivial cuts $S \subseteq V$. 

If no vertex-weighting is specified $W$ is assumed to be $deg_G$, i.e., the volume $\Vol_{G}(S)$ counts edges in $\partial_G(S)$ as one, edges induced by $S$ as two, and all other edges as zero and $\Phi_G(S)$ is the conductance with respect to this volume.  

\begin{definition}[$\phi$-expander (for $W$)]\label{def:expander}\cite{LiSWeighted}
Let $G=(V,E)$ be a graph, $W$ be a vertex weighting for $V$, and $\phi \geq 0$. $G$ is a $\phi$-expander for $W$ if $\Phi(G,W) \geq \phi$. We also say $W$ is $\phi$-expanding in $G$. If $G$ is $\phi$-expanding for $W = deg_G$ we simply say $G$ is a $\phi$-expander (leaving the "for $deg_G$" implied).
\end{definition}

\paragraph{Demands, Routings, Congestion and Dilation.}

A demand $D: V \times V \to \mathbb{R}_{\geq 0}$ assigns each pair of vertices $u,v$ a positive amount $D(u,v)$. A demand is integral if it only assigns integer amounts and symmetric if $D(u,v)=D(v,u)$ for all $u,v \in V$. All demands in this paper are symmetric and integral. The \emph{load} of a vertex $v \in V$ in demand $D$ is defined as $load(v,D)=\max\{\sum_{w \in V}D(v,w), \sum_{w \in V}D(w,v)\}$. Let (max) load of a demand is defined as $load(D)=\max_v load(v,D)$. We call a demand $W$-respecting (or $W$-unit) if $load(v,D)\leq W(v)$ and we say a demand $D$ is a unit demand in $G$ if it is $deg_G$-respecting. We say that a demand $D$ is supported on $V' \subseteq V$, it $D(u,v)=0$ for any pair $\langle u, v\rangle \notin V' \times V'$.

Any multi-set of edges $E'$ has an associated integral symmetric demand $D_{E'}$ with $D_{E'}(u,v)$ equal to the number of edges between $u$ and $v$ in $E'$. Similarly any integral symmetric demand maps to a multi-set of edges or a set of edges $E_D = \{(u,v) ~\mid~ D(u,v)>0\}$.

\begin{definition}[Dilation, Congestion of a Path Collection]
A routing is a collection of paths $\mathcal{P}=\{P_1,\ldots, P_k\}$, the \emph{dilation} of $\mathcal{P}$ is the maximum length of the paths in $\mathcal{P}$. The \emph{congestion} of $\mathcal{P}$ is equal to $\max_{e \in G}|\{P \in \mathcal{P} ~\mid~ e \in P\}|$, i.e., the maximum number of times any edge appears in the total over all paths. 
The \emph{quality} of $\mathcal{P}$ is defined as $congestion(\mathcal{P}) + dilation(\mathcal{P})$. 
The demand $D_{\mathcal{P}}$ sets $D(u,v)$ to be the number of pathes from $u$ to $v$ in $\mathcal{P}$. (Iff $D=D_{\mathcal{P}}$ we say "$\mathcal{P}$ routes $D$ or $\mathcal{P}$ is a routing for $D$.) 
\end{definition}

We say a demand $D$ is routable with congestion $\congestion$ and dilation $\dilation$ in $G$ if there exists a routing $\mathcal{P}$ for $D$ in $G$ with congestion congestion $\congestion$ and dilation $\dilation$. Similarly we say a graph $H=(V',E')$ with $V' \subseteq V$ or a set of edges $F \subseteq V \times V$ is $d$-routable with congestion $\congestion$ and dilation $\dilation$ in $G$ if $d \cdot D_H$ or $d \cdot D_F$ are routable. 
A routing for a graph $H$ in $G$ is also called an embedding (of $H$ into $G$).


We show in appendix that if a good routing or embedding exists than such a routing can also be found efficiently in polymomial time, up to a constant factor loss in congestion. This allows us to provide polynomial time implementation for our expander-based algorithms.

\begin{lemma}[Length-Constrained Multi-Commodity Flow Routing~\cite{HaeuplerHS23}]\label{lem:LC-MCF-Routing}

There exists a randomized algorithm which given an $n$-vertex graph $G=(V,E)$ and a demand $D$ which is routable with congestion $\congestion = \omega(\log n)$ and dilation $\dilation$ in $G$ outputs a routing for $D$ with congestion at most $1.1 \congestion$ and dilation $\dilation$ in polynomial time, with high probability.
\end{lemma}
\def\APPENDPOLYROUTE{
\begin{proof}[Proof of Lemma \ref{lem:LC-MCF-Routing}]
We use the multi-commodity $(1+\eps)$-approximate length-constrained flow algorithm of \cite[Theorem A.1]{HaeuplerHS23} together with the flow-path decomposition algorithm~\cite[Theorem 10.1]{HaeuplerHS23} followed by a simple independent randomized rounding. We set $\eps=\frac{1}{100mn\congestion}$ and define the capacity of all edges in $G$ to be $\congestion$. Since our flow size is always integral and of size at most $mn\congestion$ this guarantees that the fractional flow routed by the $(1+\eps)$-approximation algorithm looses at most $1/100$ units of flow overall and therefore (fractionally) routes demand of every demand pair up to an additive $1/100$ (or multiplicative factor of $1.01$). We do not use any batching structure of the multi-commodity demands and simply make each demand pair its own batch, i.e., $\kappa \leq n^2$. The running time is polynomial in $\eps$, $\kappa$, $\dilation$, $\congestion$ and $|D|$ and therefore polynomial overall. The algorithm produces a fractional multi commodity flow in the form of an average of polynomially many integral flows. We compute the flow-path decomposition for each of these integral flows using~\cite[Theorem 10.1]{HaeuplerHS23}. Overall this gives for every demand pair $u,v$ with $D(u,v)>0$ an explicit distribution $P_{u,v}$ with polynomial support over flow paths between $u$ and $v$ of length at most $\dilation$. Without loss of generality all routing paths are simple. We now independently sample exactly $D(u,v)$ paths from $P_{u,v}$ for every $u,v$ to obtain a routing for $D$ with dilation at most $\dilation$. For any edge in $G$ the congestion of these sampled paths is at most $1.01 \congestion$ in expectation. Since every (simple) path creates at most one unit of congestion per edge if sampled and since all paths are sampled independently a multiplicative Chernoff bound guarantees that the probability of any fixed edge $e$ having congestion more than $1.1 \congestion$ is at most $\exp(-\Theta(\congestion)) = n^{-\omega(1)}$. Therefore, with high probability, this randomized rounding of the fractional multi-commodity flow has congestion at most $1.1 \congestion$. 
\end{proof}
}

The next well-known lemma summarizes the routing properties of $\phi$-expanders. 
\begin{lemma}\label{lem:routing-implies-expansion}
If a graph $G$ is a $\phi$-expander, then any unit demand in $G$ can be routed with congestion $O(\frac{\log n}{\phi})$ and dilation $O(\frac{\log n}{\phi})$. 
\end{lemma}


For a graph $G=(V,E)$ and a vertex subset $U \subseteq V$, let $G[U]$ be the graph $G$ induced on $U$. The graph $G(U)$ consists of all edges of $G[U]$ and in addition, it includes (potentially multiple) self-loops on each vertex $u \in U$, such that $\deg_{G(U)}(u)=\deg_G(u)$. I.e., for each edge $(u,v)\in U \times (V \setminus U)$, the graph $G(U)$ includes a self-loop on $u$.   

\begin{theorem} [\cite{LiSWeighted}] \label{thm:expander-decomp}
For any any parameter $\phi>0$, there is a randomized algorithm
$\ExpDecomp$ that given an $n$-vertex graph $G=(V,E)$ partitions $V$ into \emph{clusters} $U=(U_1,\ldots, U_k)$, such that w.h.p. the following holds:
\begin{itemize}
\item $\sum_{i} |E(U_i, V\setminus U_i)|=O(\phi m \cdot \log^{a} n)$ for some constant $a\geq 2$, and
\item Every $G(U_i)$ is a $\phi$-expander. 
\end{itemize}
\vspace{-3pt} The running time of the algorithm is $\widetilde{O}(|E(G)|/\phi^2)$.
\end{theorem}

\paragraph{Length-Constrained Expanders.} 

Length-constrained expanders are generalizations of expanders that give more and separate control over the length of routing paths. 
We say a demand is $h$-length if $D(u,v)>0$ implies that $\dist_G(u,v)\leq h$.  The following description of length-constrained expanders follows the summary in \cite{HHG22}: The general notion of $(h,s)$-length $\phi$-expanders where $h$ is called \emph{length}, $s$ is called \emph{length slack}, and $\phi$ is called \emph{conductance} is defined in terms of so-called ``moving cuts'' and can be found in \cite{HaeuplerR022}. The exact details of their definition will not be important for us here.

All we will need is that up to a constant length factor and a poly-logarithmic congestion factor, $(h,s)$-hop $\phi$-expanders for unit-demands are equivalent to the notion of routing any $h$-length unit-demand via $(s \cdot h)$-length paths with congestion roughly $\frac{1}{\phi}$. This is made precise by the following theorem which is the equivalent of \Cref{lem:routing-implies-expansion}:

\begin{theorem}
[Flow Characterization of Length-Constrained Expanders (Lemma 3.16 of \cite{HaeuplerR022})]\label{thm:flow character} For any $h \geq 1$, $\phi < 1$, (vertex weighting $W$), and $s \geq 1$ we have that:
\begin{enumerate}
\item If $G$ is an $(h,s)$-length $\phi$-expander (for $W$), then every $h$-hop ($W$-)unit-demand can be routed in $G$ along $(s \cdot h)$ long paths with congestion $O(\log(n)/\phi)$.
\item If $G$ is not an $(h,s)$-hop $\phi$-expander (for $W$), then some $h$-hop ($W$-)unit demand cannot be routed in $G$ along $(\frac{s}{2}\cdot h)$ long paths with congestion $1/2\phi$.
\end{enumerate}
\end{theorem}



\begin{theorem} [\cite{HaeuplerR022,HHG22}] \label{thm:lc-expander-decomp}
For any graph $G=(V,E)$, conductance $\phi>0$, length bound $h$, length slack $s > 2$ and vertex weighting $W$ there exists a set of edges $C \subset E$ of size at most $O(h \phi |W| n^{s^{-0.1}} \log n)$ such that $G \setminus C$ is an $(h,s)$-length $\widetilde{\Omega}(\phi)$-expander with respect to $W$.
\end{theorem}
Our use of LC-expanders w.r.t weighting $W$ is mainly for the purpose of using the above mentioned properties of expander decomposition. That is, we enjoy the fact that $|C|$ depends on $|W|$, rather than on $|E(G)|$, to provide a cleaner algorithmic description in Sec. \ref{sec:FDspanners}.

There also exists a randomized algorithm to compute the expander decomposition guaranteed by \Cref{thm:expander-decomp} with high probability and in polynomial time~\cite{HaeuplerR022,HHG22,HHT23}.

\paragraph{Neighborhood Cover.} For a graph $G=(V,E)$ and a positive integer $r$, the $r$-neighborhood of a vertex $v$ is defined by $N_{r,G}(u)=\{ v ~\mid~ \dist_G(u,v) \leq r\}$. Given a graph $G$ and parameters $\beta,r$, an $(r,\beta)$-neighborhood cover is a collection of subgraphs $G_1,\ldots, G_k$ with the following properties: (i) the diameter of each $G_i$ is $O(\beta \cdot r)$, (ii) for every $v \in V(G)$, there is a subgraph $G_i$ that contains the entire $r$-neighborhood of $v$, i.e., $N_{r,G}(u) \subseteq G_i$ and (iii) each vertex appears in at most $\beta \cdot n^{1/\beta}$ subgraphs. 

\begin{lemma} [\cite{AwerbuchBCP98}] \label{lem:neighborhood-cover}
Given a graph $G$ and parameters $r,\beta$, there is an polynomial time algorithm $\NeighborCover$ that computes the $(r,\beta)$-neighborhood cover of $G$. 
\end{lemma}

\smallskip
\textbf{Roadmap.} In Section \ref{sec:FT-certificate}, we provide the construction of FD certificates. In Section \ref{sec:FT-spanners}, we provide constructions of FD $O(1)$-spanners that are based on LC-expanders. Then, in Section \ref{sec:exist-FD-spanner}, we provide polynomial time algorithms for computing FD $t$-spanners with nearly optimal density (conditioned on the girth congestion). Finally, lower bounds arguments are provided in Section \ref{sec:LB}.

\section{Sparse Connectivity Certificates against Bounded-Degree Faults}\label{sec:FT-certificate}

In this section we provide an algorithm for computing sparse FD-certificates and establish Theorem \ref{thm:introccmain}. In fact, our construction will provide an FD $\widetilde{O}(1)$-spanners. The structure of this section is as follows. First, we show that expanders are rigid to bounded-degree faults. Then we present an algorithm for computing rigid FD certificates for expander graphs, and finally we present an algorithm for computing competitive FD certificates for general graphs. 

\paragraph{Rigidity of Expanders to Bounded Degree Faults.}
We show that $\phi$-expanders with minimum degree $\widetilde{\Omega}(\faultdeg/\phi)$ remain good expanders under the failing on any $\faultdeg$-degree edge set.

\begin{proof}[Proof of  \Cref{thm:expander-d-deg}]
Fix a cut $(S, V \setminus S)$ in $G$ and let $\Vol_G(S)\leq \Vol_G(V \setminus S)$. Then, by the definition of expanders we have $|\partial_G(S)| \geq \phi \cdot|S|\cdot \faultdeg'$. Since each vertex in $S$ is incident to at most $\faultdeg$ edges in $F$, we have that 
$|\partial_{G\setminus F}(S)| \geq |\partial_G(S)|-|S|\faultdeg \geq (\phi \cdot|S|\cdot \faultdeg')/2~.$
Moreover, $\Vol_G(V') \geq \Vol_{G \setminus F}(V')$ for any subset $V' \subseteq V$. We therefore have that:

\begin{eqnarray*}
\Phi_{G\setminus F}(S) &=&\frac{|\partial_{G\setminus F}(S)|}{\min\{\Vol_{G\setminus F}(S), \Vol_{G\setminus F}(V \setminus S)\}} \geq 
\frac{|\partial_G(S)|-|S|\faultdeg}{\min\{\Vol_{G}(S), \Vol_{G}(V \setminus S)\}} 
\\&=&
\frac{|\partial_G(S)|}{\min\{\Vol_{G}(S), \Vol_{G}(V \setminus S)\}}-\frac{|S|\faultdeg}{\min\{\Vol_{G}(S), \Vol_{G}(V \setminus S)\}} 
\\& \geq &
\Phi_{G}(S) -\frac{|S|\faultdeg' \cdot \phi}{2\min\{\Vol_{G}(S), \Vol_{G}(V \setminus S)\}} \geq
\Phi_{G}(S) -\frac{|S|\faultdeg' \cdot \phi}{2|S|\faultdeg'}
\\&=& \Phi_{G}(S)-\phi/2 \geq \phi/2~.
\end{eqnarray*}
\end{proof}

In Appendix \ref{sec:missing}, we prove a more delicate robustness argument which holds for any graph $G$ (even not expanders) provided that $G$ admits \emph{non-trivial} routing solutions for $\faultdeg$-degree subsets of edges $F$, i.e., in which the output flow uses edges not in $F$. 

\begin{lemma}\label{lem:expander-d-deg-routing}[From Low-Congestion Routing to Fault-Tolerance]
Consider an $n$-vertex graph $G$ with minimum degree $\faultdeg'$ and let $F$ be a $\faultdeg$-degree subset of edges such that $F$ is $d$-routable in $G$ for $d=\faultdeg'/\faultdeg$ with congestion of $\InitialCong< d$. Then, $G \setminus F$ is connected.
\end{lemma}

\def\APPENDDAG{
\paragraph{Proof of Lemma \ref{lem:expander-d-deg-routing}.}
%
The solution for the routing instance $(F, d)$ provides each $(u,v)\in F$ with a multi-set of $d$ many $u$-$v$ paths $\mathcal{P}(u,v)$. Since $\InitialCong<d$, it holds that $\mathcal{P}(u,v)$ contains at least one $u$-$v$ path $P$ that does not contain $e$, for every $e=(u,v)\in F$. 
We use the collection of multi-paths $\mathcal{P}=\{\mathcal{P}(u,v)\}_{(u,v)\in F}$ to compute a directed multi-graph $D$ which captures the dependencies between the edges in $F$. Specifically, $V(D)=F$ and the multiplicity of a directed edge $(e=(u,v),e')$ equals to the number of paths in $\mathcal{P}(u,v)$ that go through $e'$. In the final multi-graph $D$, we include only edges with positive multiplicity. We refer to the digraph $D$ as the dependency graph of the output multicommodity solution. Note that the multiplicity of an edge is bounded by $\min\{\InitialCong,d\}=\InitialCong$. We start by observing the following:

\begin{observation}\label{obs:eliminate-dcycles}
The output solution $\mathcal{P}=\{\mathcal{P}(u,v)\}_{(u,v)\in F}$ with congestion $\InitialCong$ can be transformed into an alternative solution 
$\mathcal{P}_0=\{\mathcal{P}_0(u,v)\}_{(u,v)\in F}$ with congestion $c(\phi)$ whose dependency (multi) graph $D_0$ is a DAG.
\end{observation}
\begin{proof}
The procedure eliminates directed cycles in an iterative manner. Consider some directed cycle $(e_1,e_2)\circ (e_2,e_3)\circ \ldots \circ (e_{\ell-1},e_{\ell}) \circ (e_{\ell},e_1)$ in $D$. Let $e_i=(u_i,v_i)$. Each edge $(e_i,e_{i+1})$ is identified with a unique path $P_i \in \mathcal{P}(u_i,v_i)$ such that $e_{i+1}\in P_i$. We then modify the routing solution as follows: For every $i \in \{1,\ldots, \ell\}$, modify $\mathcal{P}(u_i,v_i)$ by omitting $P_i$ and adding $\{e_i\}$. It is easy to see that the congestion of each edge is preserved and moreover, all edges of the directed cycle are eliminated in the dependency graph of the resulting modified solution. This step can then be repeated until all directed cycles are removed. 
\end{proof}

We next claim the following lemma which shows that the routing solution $\mathcal{P}_0=\{\mathcal{P}_0(u,v)\}_{(u,v)\in F}$ can be further modified into $\mathcal{P}^*=\{\mathcal{P}^*(u,v)\}_{(u,v)\in F}$ with an \emph{empty} dependency graph, while maintaining the same congestion bound of $c(\phi)$. 
\begin{lemma}\label{lem:eliminate-DAG}
There is a routing solution $\mathcal{P}^*=\{\mathcal{P}^*(u,v)\}_{(u,v)\in F}$ with congestion at most $\InitialCong$ and such that its corresponding dependency graph is empty.
\end{lemma}
\begin{proof}
We employ an iterative procedure where in each iteration $i\geq 1$, given is a routing solution $\mathcal{P}_{i-1}=\{\mathcal{P}_{i-1}(u,v)\}_{(u,v)\in F}$ that corresponds to a DAG $D_{i-1}$. By Obs. \ref{obs:eliminate-dcycles}, this holds for $i=1$. The output of the $i^{th}$ iteration is a routing solution $\mathcal{P}_{i}=\{\mathcal{P}_{i}(u,v)\}_{(u,v)\in F}$ with congestion of $\InitialCong$ and a dependency graph $D_i \subset D_{i-1}$. In particular, in each iteration $i$ we eliminate one copy of some edge in $D_{i-1}$.  Hence, within a finite number of steps $\ell$, we get a routing solution $\mathcal{P}^*=\mathcal{P}_{\ell}$ with congestion $c(\phi)$ and an empty dependency graph $D_\ell$. 

We now describe the $i^{th}$ iteration. Let $e'=(u',v')$ be some sink vertex in $D_{i-1}$ and let $e=(u,v)$ be such that $(e,e') \in D_{i-1}$. Our goal is to output $\{\mathcal{P}_{i}(u,v)\}_{(u,v)\in F}$ with congestion $c(\phi)$ and $D_{i}=D_{i-1}\setminus \{(e,e')\}$. I.e., we will reduce the multiplicity of the edge $(e,e')$ by exactly one in $D_i$. 
Let $P \in \mathcal{P}_{i-1}(u,v)$ be such that $e' \in P$. Since $(e,e')\in D_{i-1}$ such a path $P$ must exist. In addition, let $P' \in \mathcal{P}_{i-1}(u',v')$ be such that $P' \cap F=\emptyset$. Since $c(\phi)<d$ and since $e'$ is a sink, such a path $P'$ exists, as well.  We use $P,P'$ to define the solution $\{\mathcal{P}_{i}(u,v)\}_{(u,v)\in F}$ as follows:
\begin{itemize}
\item Initially, let $\mathcal{P}_{i}(x,y)=\mathcal{P}_{i-1}(x,y)$ for every $(x,y)\in F$.

\item Replace $P$ from $\mathcal{P}_{i}(u,v)$ with the path $P''=P[u,u']\circ P' \circ P[v',v]$.

\item Replace $P'$ from $\mathcal{P}_{i}(u',v')$ with $\{(u',v')\}$.
\end{itemize}
Observe that the congestion bound is preserved, hence $\mathcal{P}_{i}=\{\mathcal{P}_{i}(u,v)\}_{(u,v)\in F}$ has congestion of $\InitialCong$. Note that $e'$ is still a sink in $D_{i-1}$. More specifically, we have that $D_i=D_{i-1} \setminus \{(e,e')\}$. 
%
%
Since each iteration omits reduces the multiplicity of one edge in the dependency graph, within a finite number of steps, $\ell$, we obtain a routing solution $\{\mathcal{P}_{\ell}(u,v)\}_{(u,v)\in F}$ with congestion of $\InitialCong$ and an empty dependency graph, as required. The lemma follows.
\end{proof}
\noindent We are now ready to complete the proof of Lemma \ref{lem:expander-d-deg-routing}.
By Lemma \ref{lem:eliminate-DAG}, there is an output solution $\mathcal{P}^*$ with congestion of $\InitialCong$ and with an empty dependency graph. Since $\InitialCong<d$, we get that every $e=(u,v)$, it holds that $\mathcal{P}^*(u,v)\in \mathcal{P}^*$ has a $u$-$v$ path $P$ such that $e=(u,v)\notin P$. Moreover, since the dependency graph is empty, it also holds that $P \cap F=\emptyset$. Therefore, every for every $(u,v) \in F$, we get an alternative $u$-$v$ path in $G \setminus F$, and therefore $G \setminus F$ is connected. 
}



\paragraph{Expander Sparsfication.} Our certificate computation is based on our ability to sparsify a possibly dense expanders while preserving their fault-tolerant properties w.r.t. their minimum degree. The following lemma summarizes the properties of our sparsification procedure:

\begin{theorem} [Expander Sparsification] \label{thm:expander-sparsification}
For any $\phi \in [0,1]$ and $\faultdeg' \in [\Omega(\log n),n]$, there is an algorithm $\Sparsify$ that given an $n$-vertex $\phi$-expander $G$ with min-degree $\faultdeg'$ outputs a subgraph $H \subseteq G$ such that:
\begin{itemize}
\item $H$ is an $\widetilde{\Omega}(\phi^2)$-expander.
\item $|E(H)|=O(\faultdeg' \cdot n \cdot \log n/\phi)$, and
\item the minimum degree of $H$ is $\widetilde{\Omega}(\faultdeg'\cdot \phi)$. 
\end{itemize}
\end{theorem}

We need the following auxiliary claim. We say that a graph $H \subseteq G$ is a \emph{minimal} graph that embeds a graph $\widehat{G}$ into $G$ if $H$ is the union of the $G$-paths that embeds $\widehat{G}$ into $G$.
The proof of the next lemma is provided in appendix.

\begin{lemma}\label{lem:from-embedding-to-expansion}
Let $H=(V,E_H) \subseteq G$ be a \emph{minimal} graph that embeds a $\faultdeg$-regular $(1/\log n)$-expander $\widehat{G}=(V(H),\widehat{E})$ into $G$, such that the collection of embedded paths have congestion and dilation of $\widetilde{O}(1/\phi)$. Then $H$ is a $\widetilde{\Omega}(\phi^2)$-expander.
\end{lemma}
\def\APPENDLEMEMBEXP{
\begin{proof}[Proof of Lemma \ref{lem:from-embedding-to-expansion}]
In this proof and following \cite{HaeuplerR022}, we define demand between edges for convenience. For any edge $e \in E$ and any demand $D:E \times E \to \mathbb{R}_{\geq 0}$, the load of $D$ on $e$ is the maximum between the demand sent or received by $e$ in $D$. I.e., $load(e,D)=\max\{\sum_{e' \in E}D(e,e'), \sum_{e' \in E}D(e,e')\}$. 
Let $load(D)=\max_{e \in E}load(D,e)$. The demand is called unit if $load(D)\leq 1$. By Lemma 3.16 of \cite{HaeuplerR022} (Lemma \ref{lem:routing-implies-expansion} for the edge-demand setting), it is sufficient to show that any unit-demand instance $I \subseteq E_H \times E_H$ is routable in $H$ with congestion and dilation of $O(\log n/\phi)$

Consider a unit-demand edge-instance $I_H \subseteq E_H \times E_H$ where each pair $(e,e')\in I_H$ is weighted by a demand function $D_H: E(H) \to \mathbb{R}_{\geq 0}$. We then have that 
$$load(e,D_H)=\sum_{e' \in E(H)}\max\{D_H(e,e'),D_H(e',e)\}\leq 1, \forall e \in E_H~.$$ 
%

We translate $I_H$ into a unit-demand instance $I_{\widehat{G}} \subseteq \widehat{E} \times \widehat{E}$, as follows. Let $\mathcal{P}=\{P(e) \subseteq H ~\mid~ e \in \widehat{E}\}$ be the collection of paths that embed $\widehat{G}$ into $H$. Since $H$ is minimal, we have that $E(H)=\bigcap_{P \in\mathcal{P}}E(P)$.  For every $e \in H$, let $M(e)=\{e' \in \widehat{E} ~\mid~ e \in P(e')\}$ be the subset of edges in $\widehat{E}$ such that $e$ appears on their corresponding paths in $\mathcal{P}$. By the definition of $H$, we have that $M(e)\neq \emptyset$ for every $e \in E_H$. For each edge $e \in I_H$, we choose one such edge $f(e)=e'$ for some arbitrary $e' \in M(e)$. 

Note that each edge $e' \in \widehat{G}$ is associated with at most $\ell=\widetilde{O}(1/\phi)$ edges in $H$. Define a routing instance $I_{\widehat{G}}=\{(f(e), f(e')) ~\mid~ (e,e')\in I_H\}$ and let $D_{\widehat{G}}(f(e),f(e'))=d_H(e,e')$. The load of every edge $e' \in \widehat{E}$ can be bounded by:
$$load(e',D_{\widehat{G}})\leq \sum_{e \in P(e')}load(e, D_H)\leq \ell.$$

Let $\widehat{G}_{\ell}=(V, \widehat{E},c)$ be a capacitated graph in which each edge $e'$ of $\widehat{G}$ has a capacity of $c(e')=\ell$. Clearly, $\widehat{G}_{\ell}$ is a $(1/\log n)$-expander, as well. Note also that $I_{\widehat{G}}$ is a unit-demand instance in $\widehat{G}_{\ell}$, and hence it is routable in $\widehat{G}_{\ell}$ with congestion and dilation $\widetilde{O}(1)$. Since each edge in $\widehat{G}_{\ell}$ has capacity of $\ell$, the output path collection $\mathcal{P}'$ of this routing instance has congestion of $\widetilde{O}(\ell)$ and dilation $\widetilde{O}(1)$. Each path $P'=[e_1,\ldots, e_k]$ in $\mathcal{P}'$ then translates into an $H$-path $P''=P(e_1)\circ \ldots \circ P(e_k)$ that contains a path between the edge pair $e,e' \in E(H)$ for every $(e,e')\in I_H$ such that $f(e)=e_1$ and $f(e')=e_k$. The length of these paths is bounded by $\widetilde{O}(1/\phi)$ and the congestion is bounded by $\widetilde{O}(\ell/\phi)=\widetilde{O}(1/\phi^2)$. 
\end{proof}
}

We are now ready to provide the proof of Thm. \ref{thm:expander-sparsification}.
\begin{proof}[Proof of Thm. \ref{thm:expander-sparsification}]
Let $\widehat{G}=(V(G),\widehat{E})$ be an $\faultdeg'$-regular $(1/\log n)$-expander. We define a demand $D_{\widehat{E}}$ where $D_{\widehat{E}}((u,v))=1$ for every $(u,v)\in \widehat{E}$ and $D_{\widehat{E}}(u,v)=0$, otherwise. 

Since $\widehat{G}$ is $\faultdeg'$-regular, $D$ is unit-demand in $G$, and can therefore can be solved in $G$ with congestion and dilation of $O(\log n/\phi)$ by Lemma \ref{lem:routing-implies-expansion}. Moreover, the collection of routing paths $\mathcal{P}$ can be computed in polynomial time using Lemma \ref{lem:LC-MCF-Routing}.
The output subgraph $H$ consists of the union of all paths $P \in \mathcal{P}$. 

We next show that $H$ satisfies the desired properties. Since $H$ is a minimal graph for embedding $\widehat{G}$ into $G$, by Lemma \ref{lem:from-embedding-to-expansion}, we have that $H$ is a $\widetilde{\Omega}(\phi^2)$-expander. 
By construction, we also have that $|E(H)|=O(\faultdeg' \cdot n \cdot \log n/\phi)$, as $H$ consist of $\faultdeg' n$ paths, each of length at most $O(\log n/\phi)$. It remains to bound the minimum degree in $H$. Consider a vertex $u$ and its collection of $\faultdeg$ neighbors in $\widehat{G}$. Since the congestion of the paths in $\mathcal{P}$ is $O(\log n/\phi)$, we get that $u$ is incident to at least $\faultdeg=\widetilde{\Omega}(\faultdeg'\cdot\phi)$ distinct edges on the output $u$-$v$ paths of the routing instance, which implies that $u$ is incident to at least $\faultdeg$ neighbors in $H$.
\end{proof}

\paragraph{The Algorithm for $\faultdeg$-Faulty-Degree Certificates.}

We use the following procedure $\MinDegree$ that given a graph $G$ and integer $\faultdeg'$, outputs a subgraph $G' \subseteq G$ with minimum degree $\faultdeg'$ such that $|G \setminus G'|=O(\faultdeg' n)$.

\begin{lemma}\label{lem:min-deg-maker}
Given an $n$-vertex graph $G$ and integer $\faultdeg'$, there is an algorithm $\MinDegree$ that computes a subgraph $G' \subseteq G$ with minimum degree $\faultdeg'$ and such that $|G \setminus G'|=O(\faultdeg' n)$.
\end{lemma}
\begin{proof}
Letting $G'=G$, the algorithm iteratively omits a vertex with degree at most $\faultdeg'$ in the remaining graph $G'$ until no such vertex exists. The algorithm omits at most $\faultdeg' n$ edges from $G$ which proves the desired bound. 
\end{proof}

We are now ready to present Algorithm $\FDCertificate$ which computes the desired $\faultdeg$-FD certificate $H$ with $\widetilde{O}(\faultdeg \cdot n)$ edges. The algorithm has $\ell=O(\log n)$ iterations. In each iteration $i$, the algorithm gets as input a subgraph $G_i \subseteq G$ with $\Omega(\faultdeg' n)$ edges, where initially $G_1=G$. The  output of the iteration is given by a subgraph $H_i$ and a subgraph $G_{i+1}\subseteq G_i$ to be provided as input to the next iteration $i+1$. The final output subgraph $H$ is then defined by $H=\bigcup_{i=1}^\ell H_i$.

We now focus on iteration $i \geq 1$. The algorithm starts by applying Procedure $\MinDegree$ of Lemma \ref{lem:min-deg-maker} on $G_i$ with degree threshold $\faultdeg'=\widetilde{O}(\faultdeg)$. Let $G'_i$ be the output subgraph of this procedure. Next, the algorithm applies expander decomposition on $G'_i$ which partitions $G$ into vertex-disjoint $\phi$-expanders $G_{i,1}, \ldots, G_{i,k_i}$ for $\phi=1/(\log n)^{a+1}$ and a collection of at most $m/2$ inter-expander edges, where $a$ is the constant in Theorem \ref{thm:expander-decomp}. 

On each sufficiently dense $\phi$-expander $G_{i,j}$ (i.e.,  (with $|G_{i,j}|\geq 2\faultdeg' \cdot |V(G_{i,j}|)$), the algorithm applies Procedure $\Sparsify$ of Thm. \ref{thm:expander-sparsification} to obtain a $\widetilde{\Omega}(\phi^2)$-expander $H_{i,j}\subseteq G_{i,j}$. The output of the $i^{th}$ phase is given by a subgraph $H_i=\bigcup_{i=1}^{k_i}H_{i,j}$ and the next graph $G_{i+1}$ is the defined by the union of all inter-expander edges. As $|G_{i+1}|\leq |G_i|/2$ the process terminates within $\ell$ iterations. See a detailed description below.

\begin{mdframed}[hidealllines=false,backgroundcolor=gray!30]
\center \textbf{Algorithm $\FDCertificate$}
\begin{flushleft}
Input: A graph $G$, integer $\faultdeg$. \\
Output: An $\faultdeg$-FD certificate $H \subseteq G$ with $\widetilde{O}(\faultdeg \cdot n)$ edges.
\end{flushleft}
\begin{itemize}
\item $G_1 \gets G$ and $\phi=1/\log^{a+1} n$ (see Theorem \ref{thm:expander-decomp}).
\item $\faultdeg'\gets \lceil \faultdeg /\phi^5 \rceil$ and $\faultdeg''\gets \lfloor\faultdeg'\cdot \phi \rfloor$.
\item For $i=\{1,\ldots, \ell=2\log n\}$ do:
\begin{enumerate}
\item $G'_i=\MinDegree(G_i, \faultdeg')$.
\item $(U_{i,1},\ldots, U_{i,k_i})=\ExpDecomp(G'_i, \phi)$.
\item For every $j \in \{1,\ldots, k_i\}$ do:
\begin{itemize}
\item Let $G_{i,j}=G'_i[U_{i,j}]$.
\item If $|G_{i,j}|\leq 2\faultdeg' \cdot |U_{i,j}|$: $H_{i,j} \gets G_{i,j}$.
\item Otherwise: $H_{i,j} \gets \Sparsify(G_{i,j}, \faultdeg'')$.
\end{itemize}
\item $G_{i+1} \gets \bigcup_{i} G \setminus \bigcup_{j=1}^{k_i} G_{i,j}$.
\item $H_i\gets (G_i \setminus G'_i) \cup \bigcup_{j=1}^{k_i} H_{i,j}$.
\item If $|G_{i+1}|\leq \faultdeg' n$, set $H_{i+1}=G_{i+1}$, $\ell=(i+1)$ and QUIT.
\end{enumerate}
\item $H= \bigcup_{i=1}^{\ell} H_i$.
\end{itemize}
\end{mdframed}


\begin{lemma}\label{lem:reduce-half}
$H$ is an $\faultdeg$-FD $\widetilde{O}(1)$-spanner (and hence also an $\faultdeg$-FD certificate) with $|E(H)|=\widetilde{O}(\faultdeg \cdot n)$. 
\end{lemma}
\begin{proof}
We start with the size bound by showing that $|H_i|=\widetilde{O}(\faultdeg \cdot n)$. By Lemma \ref{lem:min-deg-maker}, we have $|G_i \setminus G'_i|=\widetilde{O}(\faultdeg n)$. By Thm. \ref{thm:expander-sparsification}, $|H_{i,j}|=\widetilde{O}(\faultdeg \cdot |V(U_{i,j})|)$. The size bounds follows that $H_{i,1},\ldots, H_{i,k_i}$ are vertex-disjoint.  We next claim that the algorithm quits within $2\log n$ iterations. This holds as by the stopping criteria, and by the fact that $|G_{i+1}|\leq |G_i|/2$ (as guaranteed by Thm. \ref{thm:expander-sparsification}).

\smallskip
\noindent \textbf{Min-Degrees.} We next claim that the minimum-degree of any $\phi$-expander $G_{i,j}=G'_i[U_{i,j}]$ is at least $\faultdeg''$. 
By the properties of expander-decomposition, we have that the graph $\tilde{G}=G'_i(U_{i,j})$ is also a $\phi$-expander. Recall that the graph $\tilde{G}$ consists of all edges of $G_{i,j}$, and in addition, each vertex $v$ has a self-loop for every edge $(u,v)\in G'_i \setminus G_{i,j}$. 
By adding these self-loops, we have that 
the degree of $v$ in $\tilde{G}$ is the same as its degree in $G'_i$. Since the minimum degree of $G'_i$ is at least $f'$, we have that $\min\{\Vol_{\tilde{G}}(\{v\}),\Vol_{\tilde{G}}(V \setminus \{v\})\}\geq f'$ for every $v \in U_{i,j}$. By the expansion of $\tilde{G}$, for every $v \in U_{i,j}$, it then holds that
$\deg_{G_{i,j}}(v)\geq \phi \cdot f' \geq f''$.


\smallskip
\noindent \textbf{Stretch.} Fix an $\faultdeg$-degree set $F$ and $(u,v) \in E \setminus F$. The interesting case is when $(u,v)\notin H$, and therefore in such a case there exist indices $i,j$ such that $(u,v)\in G_{i,j}\setminus H_{i,j}$. By Theorem \ref{thm:expander-sparsification}, $H_{i,j}$ is a $\widetilde{\Omega}(\phi^2)$-expander with minimum degree $\tilde{\faultdeg}=\widetilde{\Omega}(\faultdeg'' \cdot \phi)=\widetilde{\Omega}(\faultdeg' \cdot \phi^2)$.
Since $\tilde{\faultdeg}=\widetilde{\Omega}(\faultdeg/\phi^3)$, by Theorem \ref{thm:expander-d-deg}, we have that $H_{i,j} \setminus F$ is also $\widetilde{\Omega}(\phi^2)$ expander. Concluding that $\dist_{H_{i,j} \setminus F}(u,v)=\widetilde{O}(1/\phi^2)=\widetilde{O}(1)$.
\end{proof}

\section{Sparse Spanners against Bounded-Degree Faults}\label{sec:FT-spanners}

\subsection{Robustness of Length-Constrained Expanders against Bounded Degree Faults}\label{sec:robust-LCEx}

In this section, we show that dense LC-expanders $G$ are robust to bounded-degree faults $F$, in the sense that $G\setminus F$ is a $O(1)$-spanner for $G$. We show the following (a generalization of Thm. \ref{thm:bounded-hop-expander-robust}):

\begin{theorem}\label{lem:expander-stretch-robustness}
Consider an $(h,s)$-length $\phi$-expander $G$ for vertex weighting $W$ where $W(u)=\widetilde{\Omega}(\faultdeg \cdot n^{\epsilon}/\phi)$, for every $u \in V(G)$. Then, for every $\faultdeg$-degree subset of edges $F$, it holds that $\dist_{G \setminus F}(u,v)\leq (h \cdot s)^{O(1/\epsilon)}$ for every $(u,v)\in F$. 
\end{theorem}

To prove the theorem, we use the routing properties of LC-expanders, and in particular the following corollary implied by Theorem \ref{thm:flow character}:

\begin{corollary}[Routing on Dense LC Expanders]\label{lem:high-deg-expander-routing}
Let $G$ be an $n$-vertex $O(h,s)$-length $\phi$-expander for vertex weighting $W$ where $W(u)=\widetilde{\Omega}(\faultdeg \cdot n^{\epsilon}/\phi)$ for every $u \in V(G)$. Then every $\faultdeg$-degree subset of edges $F$ is $\Theta(n^{\epsilon}/\phi)$-routable in $G$ with dilation $h\cdot s$ and congestion $O(\log n/\phi)$.
\end{corollary}


By Cor. \ref{lem:high-deg-expander-routing}, it is sufficient to show the following key lemma that translates routing properties of edge sets $F \subseteq G$ into fault-tolerant properties in $G \setminus F$. Recall that a given set of edges $F$ is $d$-routable with congestion $\congestion$ and dilation $\dilation$, if one can solve a routing instance in which $d$ units of flow are required to be sent over each edge $e \in F$ along paths with congestion $\congestion$ and dilation $\dilation$. 


\begin{theorem}\label{lem:router-stretch-robustness}
Let $G$ be an $n$-vertex graph and let $F \subseteq E(G)$ be such that $F$ is $d$-routable with congestion $\congestion \leq d$ and dilation $\dilation$. Then, for every $(u,v)\in F$, it holds that $\dist_{G \setminus F}(u,v)\leq \dilation^{k}$ where $k=2\log n/ \log (d/\congestion)$. 
\end{theorem}

\begin{proof}
Set $k=2\log_{d/\congestion} n$. 
We partition $F$ into at most $k$ subsets $F_1,\ldots, F_k$, such that for every $(u,v)\in F_i$ it holds that
$\dist_{G \setminus F'_i}(u,v)\leq \dilation$ where $F'_1=F$ and $F'_i=F \setminus \bigcup_{j\leq i-1} F_j$ for $i \in \{2,\ldots, k\}$. This is done in $k$ iterations. 

In iteration $i \in \{1,\ldots, k\}$, we define a routing instance where the routing pairs are the edges in $F'_i$ and with uniform demand $D(u,v)=d$ for every $(u,v)\in F'_i$ and $D(u,v)=0$ otherwise.  By solving this routing instance in $G$, we get the routing paths $\mathcal{P}_i=\{\mathcal{P}(u,v) ~\mid~ (u,v)\in F'_i\}$ where each $\mathcal{P}(u,v)$ is a multi-set of $d$ paths connecting $u$ and $v$ in $G$. The set $F_i$ is then defined by: 
$$F_i=\{ (u,v)\in F'_i ~\mid~ \exists P \in \mathcal{P}(u,v) \mbox{~such that~} P \cap F'_i=\emptyset\}.$$ 
That is, $F_i$ consists of all edges $(u,v) \in F'_i$ whose $\mathcal{P}(u,v)$ set contains at least one path that does not intersect $F'_i$.  We next claim that $F_1,\ldots, F_k$ is a partitioning of $F$. 

We show that $|F_i|\geq (1-\congestion/d)|F'_i|$ for every $i \in \{1,\ldots, k\}$, which establishes the claim. This is can be shown by a simple averaging argument: The number of paths in $\mathcal{P}_i$ that intersects $F'_i$ is at most $\congestion \cdot |F'_i|$. Hence $|F'_i \setminus F_i|\leq \congestion \cdot |F'_i|/ d$. Consequently, $F'_{k+1}=\emptyset$ and $\bigcup_{i=1}^k F_i=F$.

Finally, we show by induction on $i \in \{1,\ldots, k\}$ that $\dist_{G \setminus F}(u,v)\leq \dilation^i$. The base case follows by the definition of $F_1$, as $F'_1=F$. Assume that the claim holds up to $(i-1)$ and consider $i$. Fix $(u,v)\in F_i$ and consider some $u$-$v$ shortest path $P$ in $G \setminus F'_i$. By the definition of $F_i$, we have that $|P|\leq \dilation$. Since $P \cap F \subseteq \bigcup_{j\leq i-1} F_j$, we can replace each edge $(x,y)\in \bigcup_{j\leq i-1} F_j$ with a fault-free $x$-$y$ path $P_{x,y}\subseteq G \setminus F$. By induction assumption we have that $|P_{x,y}|\leq \dilation^{i-1}$. Altogether, we get that $\dist_{G \setminus F}(u,v)\leq |P|\cdot \dilation^{i-1}\leq \dilation^i$.
\end{proof}

\noindent Theorem \ref{thm:bounded-hop-expander-robust} then follows by Lemma \ref{lem:router-stretch-robustness} and Lemma \ref{thm:flow character}.

\subsection{Length-Constrained Expander Sparsification}\label{sec:bhop-sparsification}

Our goal in this section is to sparsify an expander $G$ into $H \subset G$ while maintaining the congestion and dilation parameters of routing instances. Our key sparsification result shows:

\begin{theorem}[Key Sparsification Lemma] \label{thm:expander-sparsification-unit-demand}
Consider an $n$-vertex graph $G=(V,E)$ with and a subset $V' \subseteq V$ satisfying that any demand $D$ supported on $V'$ with $load(D)\leq \faultdeg'$, for some $\faultdeg'=\Omega(\log n)$, can be routed with dilation $\dilation_G$ and congestion $\congestion_G$. Then, there is a randomized polynomial-time procedure $\BoundedHopSparsify$ that computes a subgraph $H \subseteq G$ which satisfies the following properties w.h.p.:
\begin{itemize}
\item $\min_{v \in V'}\deg_{H}(v)=\Omega(\faultdeg'/\congestion_G)$.
\item The diameter of $H$ is $O(\log_{\faultdeg'} n \cdot \dilation_G)$.
\item Any demand function $D$ supported on $V'$ such that $load(D)\leq \Theta(\faultdeg')$ is routable in $H$ with congestion $\widetilde{O}(\congestion_G)$ and dilation $O(\dilation_G \cdot \log_{\faultdeg'}n)$.

\item $|E(H)|\leq \faultdeg' \cdot |V'|\cdot \dilation_G$.
\end{itemize}
\end{theorem}
In the following we describe Algorithm $\BoundedHopSparsify$ and prove Thm. \ref{thm:expander-sparsification-unit-demand}.

\paragraph{Algorithm $\BoundedHopSparsify$.} Algorithm $\BoundedHopSparsify$ that receives as input a graph $G=(V,E)$, a subset $V' \subseteq V$ and an integer parameter $\faultdeg'$. The algorithm computes $H \subseteq G$ that preserves the routing quality of $G$ on any demand $D$ supported on $V'$ with $load(D)\leq \faultdeg'$.

Our approach is based on embedding a ``good constant-hop expander" $\widehat{G}$ on the vertex set $V'$ in the graph $G$, and letting $H$ be the union of the paths that embeds $\widehat{G}$ in $G$. Here we take $\widehat{G}=(V',\widehat{E})$ to be the random graph $G(n,p)$ with $p=\Theta(f'/n)$ because random graphs with polynomially large degrees are excellent routers allowing any unit demand to be routed over constant length paths with low congestion. More precisely it is true that with high probability any unit-demand in $\widehat{G}$ is be routable with dilation $\dilation_{\widehat{G}}=O(\log_{\faultdeg'} n)$ and congestion $c_{\widehat{G}}= O(\log n)$. We prove this in \Cref{lem:gnprouting}. Also note by setting the constant hidden in the probability $p$ to be small enough, w.h.p., all degrees in $\widehat{G}$ are in $[\faultdeg'/a, \faultdeg']$ for some constant $a\geq 1$.


Next, define a routing instance in $G$ with the following demands values: 
 $D(u,v)=1$ for every $(u,v)\in \widehat{E}$ and $D(u,v)=0$, otherwise. Since $D$ is supported on $V'$ and $load(D)\leq \faultdeg'$, by the properties of $G$, the demand is routable in $G$ by a collection of paths 
$\mathcal{P}=\{P(x,y) ~\mid~ (x,y)\in \widehat{E}\}$ with dilation at most $\dilation_G$ and congestion $\congestion_G$. Moreover, these paths can be computed in polynomial time using Lemma \ref{lem:LC-MCF-Routing}.
The output subgraph $H \subseteq G$ is then given by $H=\bigcup_{P \in \mathcal{P}}E(P)$. 

\smallskip 
\noindent\textbf{Size and Diameter.}
The size bound is immediate as $|E(H)|\leq \faultdeg' \cdot |V'| \cdot \dilation_G$. We next consider the diameter bound. The diameter of $\widehat{G}$ is $O(\log_{\faultdeg'} n)$, and as each edge in $\widehat{G}$ translates into a $H$-path of length $\dilation_G$, we have that the diameter of $H$ is $O(\log_{\faultdeg'} n \cdot \dilation_G)$. 

\smallskip 
\noindent\textbf{Minimum Degree.} Since each $G$-edge appears on $\congestion_G$ paths in $\mathcal{P}$, and since each $u$ is the source vertex of $\Theta(\faultdeg')$ paths in $\mathcal{P}$, we have that $u$ is incident to $\Omega(\faultdeg'/\congestion_G)$ edges in $H$. Note that some vertices in $V(H)\setminus V'$ might have arbitrarily small or large degrees.

\smallskip 
\noindent\textbf{Routing.} Fix a demand $D$ supported on $V'$ with $load(D)\leq \faultdeg'/a$ and recall that the minimum degree of $\widehat{G}$ is at least $\faultdeg'/a$. Define $I_H=\{\langle u,v \rangle ~\mid~ D(u,v)\neq \emptyset\}$, and as $D$ is supported on $V'$, we have that all pairs $\langle u,v \rangle \in I_H$ are in $V' \times V'$. 

Since $D$ is a unit-demand instance in $\widehat{G}$, it can be routable in $\widehat{G}$ with dilation at most $\dilation_{\widehat{G}}=O(\log_{\faultdeg'} n)$ and congestion $\congestion_{\widehat{G}}=\widetilde{O}(1)$. 
Let $\mathcal{Q}=\{\mathcal{Q}(u,v) ~\mid~ \langle u,v \rangle \in I_H\}$ be the collection of output paths in $\widehat{G}$, where $\mathcal{Q}(u,v)$ is a multi-set of $D(u,v)$ paths connecting $u$ and $v$ in $\widehat{G}$. Since each edge $(x,y)$ in $\widehat{G}$ translates into a path $P(x,y)$ in $H$, the routing solution in $\widehat{G}$ can be translated into a routing solution in $H$, as follows. Letting $Q=[u=x_1,\ldots, x_k=v] \in \mathcal{Q}(u,v)$, then let
$$g(Q)=P(x_1,x_2) \circ \ldots \circ P(x_{k-1},x_k)~.$$
Then, define $\mathcal{P}'(u,v)=\{g(Q) ~\mid~ Q \in \mathcal{Q}(u,v)\}$ and $\mathcal{P}'=\bigcup_{\langle u,v \rangle \in I_H}\mathcal{P}'(u,v)$. We finally note that path collection $\mathcal{P}'$ has dilation $\dilation_G \cdot \dilation_{\widehat{G}}$ and congestion is at most $\congestion_{\widehat{G}}\cdot \congestion_G$. The lemma follows.

\subsection{Expander-Based Computation of FD-Spanners}\label{sec:FDspanners}

We are now ready to present Alg. $\FDSpanner$. The algorithm computes the FD spanner $H$ by taking the following steps. First, it computes a subgraph $G' \subseteq G$ with minimum degree $\faultdeg'=\widetilde{\Theta}(\faultdeg \cdot n^{1/t})$ using Lemma \ref{lem:min-deg-maker}. The edges in $G \setminus G'$ are added to $H$. Then, it applies the LC-expander decomposition of Theorem \ref{thm:lc-expander-decomp} with respect to expansion parameter $\phi=1/\log n$, $h=O(t), s=t^{20}$ and vertex weighting $W(u)=\faultdeg \cdot n^{1/(2t)}$ for every $u\in V$. The output of this expander decomposition is a set $C$ of cut edges such that $G''=G' \setminus C$ is a $(h,s)$-length $\phi$-expander w.r.t $W$. All cut edges $C$ are then added to $H$. Next, the algorithm applies the neighborhood-cover procedure of 
Lemma \ref{lem:neighborhood-cover} on $G''$ with radius parameter $r=O(t)$ and overlap of $\beta=O(t n^{1/r})$. 
Formally, let $(G_{1},\ldots, G_{\ell})=\NeighborCover(G'',r=1,\beta=3t)$. For every $j \in \{1,\ldots, \ell\}$, it applies the expander-sparsification procedure of Thm. \ref{thm:expander-sparsification-unit-demand} to compute a subgraph $H_j$ that preserves the routing quality of $G''$ with respect to pairs in $V[G_{j}]$. Each $H_j$ is then added to $H$. This completes the description, and the pseudocode appears below. 

\begin{mdframed}[hidealllines=false,backgroundcolor=gray!30]
\center \textbf{Algorithm $\FDSpanner$}
\begin{flushleft}
Input: A graph $G$, integer $\faultdeg$, a stretch parameter $t$. \\
Output: An $\faultdeg$-FD $t^{O(t)}$-spanner $H \subseteq G$ with $\widetilde{O}(\faultdeg \cdot n^{1+1/t})$ edges.
\end{flushleft}
\begin{itemize}
\item $G' \gets \MinDegree(G, \faultdeg')$ for $\faultdeg'=\widetilde{\Theta}(\faultdeg \cdot n^{1/t})$.
\item $\phi \gets 1/\log n$, $W(u)=\faultdeg \cdot n^{1/(2t)}$ for every $u \in V$.
\item $C=\BoundedHopExpDecomp(G', \phi, h=3t, s=t^{20},W)$.
\item $G''\gets G' \setminus C$ and $H_0\gets C$.
\item $(G_{1},\ldots, G_{\ell})=\NeighborCover(G'',r=1,\beta=5t)$.
\item For every $j \in \{1,\ldots, \ell\}$ do:
\begin{itemize}
\item $H_{j} \gets \BoundedHopSparsify(G'',V(G_{j}),\faultdeg \cdot n^{1/(2t)})$. 
\end{itemize}
\item $H \gets (G \setminus G') \cup \bigcup_{j=0}^{\ell} H_{j}$.
\end{itemize}
\end{mdframed}


\smallskip \noindent\textbf{Size Analysis.} By Lemma \ref{thm:lc-expander-decomp}, we have that $|H_0|=\widetilde{O}(h \phi \cdot n \cdot \faultdeg \cdot n^{1/(2t)} \cdot n^{1/t^2})=\widetilde{O}(\faultdeg \cdot n^{1+1/t})$.

By Theorem \ref{thm:expander-sparsification-unit-demand}, $|H_{j}|=\widetilde{O}(\faultdeg \cdot n^{1/(2t)} (h\cdot s)\cdot |V(G_{j})|)$.
In addition, by Lemma \ref{lem:neighborhood-cover}, each vertex appears in $O(t \cdot n^{1/(5t)})$ subgraphs $G_i$. 
Therefore, as $t=O(\log n)$, we have that $\sum_{j} |H_{j}|=\widetilde{O}(\faultdeg \cdot n^{1+1/t})$. Finally, by Lemma \ref{lem:min-deg-maker}, we have that $|G\setminus G'|=\widetilde{O}(\faultdeg \cdot n^{1/t})$. 
%

\smallskip 
\noindent \textbf{The Stretch Argument.} Let $F$ be an $\faultdeg$-degree faulty set. It is sufficient to show that for every $(u,v) \in E(G) \setminus F$, it holds that $\dist_{H \setminus F}(u,v)=t^{O(t)}$. Fix $(u,v) \in E(G) \setminus F$. The interseting case is clearly when $(u,v)\notin H$. This in particular implies that $(u,v)\in G''$, and thus by the properties of the neighborhood cover of Lemma \ref{lem:neighborhood-cover}, there exists $j$ such that $(u,v)\in G_j$.
Since $G''$ is a $\phi$-expander w.r.t $W$ and $W(u)=\faultdeg \cdot n^{1/(2t)}$ for every $u$, we have:

\begin{corollary}
\label{lem:obs-expander-stretch}
For every $j \in \{1,\ldots, \ell\}$, any demand $D$ supported on $V[G_j]$ with $load(D)\leq \faultdeg \cdot n^{1/(2t)}$ can be routed in $G''$ along paths of dilation at most $O(t^{21})$ and congestion $O(\log n/\phi)$.
\end{corollary}


Our goal is to show that $\dist_{H\setminus F}(x,y)\leq t^{O(t)}$ for every $(x,y) \in F$. To prove the claim, we show that (i) $H_j$ is a good router for the edges in $F_j=F \cap G_j$ for every $j \in \{1,\ldots, \ell\}$, and consequently that (ii) $H$ is a good router for $F$.

\begin{corollary}\label{obs:good-routing-F}
$F_j=F \cap G_j$ is $d$-routable in $H_j$ with dilation $O(t^{22})$ and congestion $\congestion=O(\log n/\phi)$ for $d=n^{1/(3t)}\cdot \congestion$. 
\end{corollary}
\begin{proof}
By Corollary \ref{lem:obs-expander-stretch} and Theorem \ref{thm:expander-sparsification-unit-demand}, we have that any demand $D$ supported on $V(G_j)$ of load at most $k=\Theta(\faultdeg \cdot n^{1/(2t)})$ is routable in $H_j$ with congestion of $\widetilde{O}(1/\phi)$ and dilation $O(t^{22})$. In our setting, the demands are $D(u,v)=d$ for every $(u,v) \in F_j$ and $D(u,v)=0$, otherwise. Hence, $D$ is supported on $V(G_j)$ and $load(D)\leq d \cdot \faultdeg \leq k$. Therefore, $D$ is routable in $H_{j}$ with congestion $\widetilde{O}(1/\phi)$ and dilation $O(t^{22})$. 
\end{proof}

\begin{observation}\label{obs:good-routing-F2}
$F$ is $d'$-routable in $H$ with dilation $O(t^{22})$ and congestion $\congestion'=\widetilde{O}(n^{1/(5t)})$ for $d'=n^{1/(20t)}\cdot \congestion'$. 
\end{observation}
\begin{proof}
By Corollary \ref{obs:good-routing-F}, each $F_j$ is $d$-routable in $H_j$ with dilation $O(t^{22})$ and congestion $\congestion=O(\log n/\phi)$ for $d=n^{1/(3t)}\cdot \congestion$. 

By Lemma \ref{lem:neighborhood-cover}, each vertex appears in at most $q=O(t n^{1/(5t)})$ sets of $G_{1},\ldots, G_{\ell}$, and that $\bigcup_j F_{j}=F$. We therefore get that $F$ is $d$-routable in $H$ with with congestion of 
$\congestion'=\congestion \cdot q$ and dilation $O(t^{22})$. The claim follows as $d'\leq d$.
\end{proof}

By combining Cor. \ref{obs:good-routing-F2} with Lemma \ref{lem:router-stretch-robustness}, we get that $\dist_{H\setminus F}(x,y)=t^{O(t)}$. Since $\dist_{H}(u,v)\leq \dist_{H_{j}}(u,v)=O(t)$, it holds that $\dist_{H \setminus F}(u,v)=t^{O(t)}$.

\section{Improved Existential Upper Bounds for FD Spanners}\label{sec:exist-FD-spanner}

In this section, we provide near-tight bounds for FD spanners for any fixed stretch value. In contrast to the previous constructions, the approach provided in this section is structural and does not use expanders. 
\begin{theorem} \label{thm:extremalfdspan}
For all positive integers $n, k, \faultdeg$, every $n$-node undirected weighted graph has a $\faultdeg$-FD $(2k-1)$-spanner on at most $\faultdeg^{1 - 1/k} n^{1+1/k} \cdot O(k)^{k}$ edges.
\end{theorem}

By the lower bound in the following section, this bound is optimal for constant $k$, but may be suboptimal in its $O(k)^{k}$ dependence.
We refer back to Section \ref{sec:exist-stretch} for a technical overview of the proof strategy used in this section, although some of it will be repeated here.

\subsection{Proof Setup \label{sec:exspannersetup}}

We construct our spanners using the following greedy algorithm, which is the natural adaptation to our setting of the standard greedy algorithm for non-faulty spanners \cite{AlthoferDDJS:93}, or more specifically, the greedy algorithm used in previous work on fault tolerant spanners \cite{BDPW18, BP19, BodwinDR22}.

\begin{mdframed}[hidealllines=false,backgroundcolor=gray!30]
\center \textbf{Algorithm $\GreedyFDSpanner$}
\begin{flushleft}
Input: A graph $G$ and a positive integer stretch parameter $k$. \\
Output: An $\faultdeg$-FD $(2k-1)$-spanner $H \subseteq G$ with at most $\faultdeg^{1-1/k} \cdot n^{1+1/k} \cdot O(k)^k$ edges.
\end{flushleft}
\begin{itemize}
\item Let $H \gets (V, \emptyset, w)$ be the initially-empty spanner
\item For each edge $(u, v) \in E$ in order of nondecreasing $w(u, v)$

\begin{itemize}
\item If there exists an edge set $F_{(u, v)}$ of max degree $\faultdeg$ such that $\dist_{H \setminus F}(u, v) > (2k-1) \cdot w(u, v)$
\begin{itemize}
\item Add $(u, v)$ to $H$
\end{itemize}
\end{itemize}

\item Return $H$
\end{itemize}
\end{mdframed}

We note that a naive implementation of this greedy algorithm runs in exponential time, since we would need to search over all possible choices of fault set $F_{(u, v)}$ in each round.
We are not currently aware of a speedup in this step.
The proof of correctness, that $H$ is indeed a $\faultdeg$-FD $(2k-1)$-spanner of the input graph $G$, is entirely standard.
We give this next.

\begin{lemma} [Correctness of the FD-Greedy Algorithm] \label{lem:greedyspannercorrect}
The output graph $H$ from Algorithm $\GreedyFDSpanner$ is a $\faultdeg$-FD $2k-1$ spanner of the input graph $G$.
\end{lemma}
\begin{proof}
Let $F$ be any set of faulty edges of degree $\faultdeg$, and so our goal is to show that $H \setminus F$ is a $(2k-1)$-spanner of $G \setminus F$.
As is well known (e.g.\ \cite{AlthoferDDJS:93}), it suffices to prove that for each edge $(u, v) \in E(G \setminus F)$, we have
$$\dist_{H \setminus F}(u, v) \le (2k-1) \cdot w(u, v).$$
Indeed, this inequality follows straightforwardly from the greedy algorithm.
If the greedy algorithm adds $(u, v)$ to $H$, then we have $\dist_{H \setminus F}(u, v) \le w(u, v)$.
If not, then by construction we have $\dist_{H \setminus F}(u, v) \le (2k-1) \cdot w(u, v)$ at the moment $(u, v)$ is considered by the greedy algorithm, and any further edges added to $H$ after this point can only decrease $\dist_{H \setminus F}(u, v)$.
\end{proof}

The goal of the rest of the proof is to control the number of edges $|E(H)|$ in the output spanner.
This is the more involved part of the proof.
Our bound on $|E(H)|$ will use the blocking set method, as in \cite{BP19, BodwinDR22, BDN23}.
The relevant kind of blocking set for this problem will be the following:
\begin{definition} [FD Blocking Sets]
Let $H = (V, E)$ be a graph equipped with a total ordering of its edge set $E$.
A $\faultdeg$-fault-degree (FD) $k$-blocking set $\bee$ for $H$ is a set of pairs of the form $(e, F_e)$, such that:
\begin{itemize}
\item Each edge $e$ is the first edge of exactly one pair $(e, F_e)$, each $F_e$ is a set of edges from $E(H)$ of maximum degree $\faultdeg$, and each edge in $F_e$ strictly precedes $e$ in the edge-ordering of $H$.
\item For each cycle $C$ in $H$ on $|C| \le k$ edges, letting $e$ be the latest edge in $C$ in the edge-ordering of $H$, we have that $F_e \cap C$ is nonempty.
\end{itemize}
\end{definition}

\begin{lemma} [Blocking Set Exists] \label{lem:blockingsetexists}
The output graph $H$ from the greedy algorithm, with $E(H)$ ordered by the order its edges were added in the greedy algorithm, has a $\faultdeg$-FD $2k$-blocking set $\bee$.
\end{lemma}
\begin{proof} [Proof Sketch]
Each time we add an edge $e$ in the greedy algorithm, we do so because of a possible fault set $F_e$.
(There might be many possible fault sets forcing us to add $e$, in which case we fix $F_e$ to be any one of them.)
We then add $(e, F_e)$ as a pair to the blocking set $\bee$.

Now letting $\bee$ be the final blocking set at the end of the greedy algorithm, it is immediate by construction that each edge $e \in E(H)$ is the first edge of exactly one pair in $\bee$, and that each $F_e$ has maximum degree $\faultdeg$.
For the last property, let $C \subseteq E(H)$ be a cycle on $\le 2k$ edges, and let $(u, v)$ be the latest edge in $C$ considered by the greedy algorithm.
When we add $(u, v)$, we must include at least one edge in $C$ in the associated fault set, since otherwise there is a $u \leadsto v$ path through $C$ of length $\le (2k-1) \cdot w(u, v)$.
Thus $C \cap F_e$ is nonempty.
\end{proof}

The focus of the proof now shifts.
Instead of controlling the number of edges in the specific output graph $H$ from Algorithm $\GreedyFDSpanner$, instead we will control the number of edges in \emph{any} $n$-node graph that admits a $\faultdeg$-FD $2k$-blocking set.
That is, our goal is now to prove the following lemma:
\begin{lemma} [Blocking Set Size Bounds] \label{lem:blockingsetsizebound}
For any positive integers $n, k, \faultdeg$, any $n$-node graph $H$ with a $\faultdeg$-FD $2k$-blocking set $\bee$ satisfies
$$|E(H)| \le \faultdeg^{1-1/k} n^{1+1/k} \cdot O(k)^k.$$
\end{lemma}
In the following arguments, $H$ is the input graph from this lemma, which thus has a blocking set as described.

\subsection{MUCk Paths and Dispersion Lemma \label{sec:muck}}

Our bound on $|E(H)|$ will be a counting argument over a special type of path that we will call a MUCk path, defined as follows:

\begin{definition} [MUCk Paths]
A path $\pi$ in $H$ is called a MUCk path if it satisfies the following properties:
\begin{itemize}
\item (Monotonic) The edges in $\pi$ occur in strictly increasing order in the edge-ordering of $H$.

\item (Unblocked) $\pi$ does not contain two edges $e, f$ with $e \in F_{f}$.

\item (Chain-Unblocked) For each edge $(u, v) \in \pi$, there are at most $C \faultdeg k$ edges $e$ with the property that $(u, v) \in F_e$, and also $e$ is incident on a node $x$ that strictly follows $(u, v)$ on $\pi$ (in particular $x \notin \{u, v\}$).
Here $C \ge 1$ is a parameter that we will set later.

\item ($k$-Path) $\pi$ has exactly $k$ edges.  (The monotonic property implies that these edges are all distinct.)
\end{itemize}
\end{definition}

In order to help motivate this definition, let us recall the discussion in Section \ref{sec:exist-stretch}.
Analogizing the Moore bounds, which limit the maximum possible number of edges in a high-girth graph, our proof will include two steps:
\begin{itemize}
\item A \emph{dispersion lemma}, which shows that not too many MUCk paths can share endpoints (or else they violate the existence of the FD blocking set), and
\item A \emph{counting lemma}, which shows a lower bound on the number of MUCk paths in $H$, where the bound scales with the number of edges in $H$.
\end{itemize}
Our edge bound will follow by comparing the upper and lower bounds on the number of MUCk paths that respectively arise from the dispersion and counting lemmas, and rearranging terms.

The main challenge in this proof method is to find a precise combination of properties that are specific enough to enable a dispersion lemma, but also general enough that many such paths are guaranteed to exist, so that the counting lemma applies. 
The MUCk definition is a cocktail of properties from prior work that does exactly that: monotonicity appears in \cite{bodwin2023alternate}, a version of the chain-unblocked property is used in \cite{BDN23}, and the unblocked property has been used repeatedly, e.g., \cite{BodwinDR22}.
Rather than saying more about the exact role played by each specific property, we will next show them in action by stating and proving our \emph{dispersion lemma}.
The properties can be best understood as outlining a maximally general class of paths over which the following dispersion lemma proof strategy applies.
The lemma shows that the MUCk paths must be ``dispersed'' around the graph $H$, rather than having too many of them concentrated on the same pair of endpoints.

\begin{lemma} [Dispersion Lemma]
For any two nodes $s, t$, $H$ has at most $O\left(C \faultdeg k\right)^{k-1}$ MUCk paths with endpoints $(s, t)$.
\end{lemma}
\begin{proof}
For the sake of an inductive proof, we will prove the following stronger claim: for each integer $1 \le j \le k$, $H$ has at most $O\left(C \faultdeg k\right)^{j-1}$ MUC $j$-paths with endpoints $(s, t)$.
The induction is on $j$, and the base case $j=1$ holds since $H$ has at most one $1$-path (edge) between $s$ and $t$.

For the inductive step, suppose $j \ge 2$, and let $A$ be the set of edges that are used as the last edge of at least one MUC $j$-path with endpoints $(s, t)$ (in particular, all edges in $A$ are incident on $t$).
If we have $|A| \le O(C \faultdeg k)$, then the rest of the proof is a straightforward counting argument: every MUC $j$-path can be expressed by choosing an edge $a \in A$ and a MUC $j-1$ path from $s$ to the non-$t$ endpoint of $a$, and so by applying the inductive hypothesis, the number of such paths is bounded by
$$|A| \cdot O\left(C \faultdeg k\right)^{j-2} \le O\left(C \faultdeg k\right)^{j-1}.$$

~\\

\begin{center}
\begin{tikzpicture}[node distance=2cm, auto, >=stealth,thick,
   every node/.style={scale=.8}, 
   every edge/.style={draw, very thick}]

    \node [circle, draw, minimum size=6pt] (s) {s};
    \node [right = 8cm of s, circle, draw, minimum size=6pt] (t) {t};

    \node[above left = of t, inner sep=0pt, circle, draw, minimum size=6pt, fill=black] (a) {};
    \node[left = of t, inner sep=0pt, circle, draw, minimum size=6pt, fill=black] (b) {};
    \node[below left = of t, inner sep=0pt, circle, draw, minimum size=6pt, fill=black] (c) {};

    \draw [thick] (t) -- (a) node [midway, above right] {$a_1$};
    \draw [thick] (t) -- (b) node [midway, above] {$a_2$};
    \draw [thick] (t) -- (c) node [midway, above left] {$a_3$};
    
    \path[black, ultra thick, fill=gray] (s) to[bend left=10] node [midway, below, white, rotate=15] {$O(C \faultdeg k)^{j-2}$} (a) to[bend left=10] (s) -- cycle;
    \path[fill=gray] (s) to[bend left=10] node [midway, below, white] {$O(C \faultdeg k)^{j-2}$} (b) to[bend left=10] (s) -- cycle;
    \path[fill=gray] (s) to[bend left=10] node [midway, below, white, rotate=-15] {$O(C \faultdeg k)^{j-2}$} (c) to[bend left=10] (s) -- cycle;
    
    
    \node [above right = 2cm of s, draw=black] {MUC $(j-1)$-paths};
    \node [above = 1.4cm of t, draw=black] {edges from $A$};

\end{tikzpicture}
\end{center}

So now it only remains to prove that $|A| \le O\left(C \faultdeg k \right)$.
Assume for contradiction that $|A| \ge 6C\faultdeg k$, let $A_{\text{late}} \subseteq A$ be the latest $4C\faultdeg k$ edges of $A$ in the ordering, and let $A_{\text{early}} \subseteq A$ be the earliest $2C\faultdeg k$ edges of $A$ in the ordering.
For each edge $a \in A_{\text{late}}$, by the fault-degree constraint there are at most $\faultdeg$ edges in $F_a$ incident to $t$, and therefore there are at most $\faultdeg$ edges in $F_a \cap A_{\text{early}}$.
Thus (somewhat conservatively) there are at least $C \faultdeg k$ edges in $A_{\text{early}}$ that are \emph{not} also in $F_a$.
We can therefore count that there exists an edge $a' \in A_{\text{early}}$ with the property that there are more than $C \faultdeg k$ edges $a \in A_{\text{late}}$ with $a' \notin F_a$.

Let us analyze this particular edge $a'$.
Let $\pi'$ be a MUC $j$-path with endpoints $(s, t)$ that uses $a'$ as its last edge.
Let $a \in A_{\text{late}}$ be one of the edges described above with $a' \notin F_a$, and let $\pi^a$ be a MUC $j$-path with endpoints $(s, t)$ that uses $a$ as its last edge.
Notice that $\pi' \cup \pi^a$ contains a cycle $C$ on $|C| \le 2k$ edges, and by monotonicity, the latest edge in $C$ is $a$.
Thus, by definition of FD blocking sets, there exists an edge $e \in C \cap F_a$.

Let us now consider the placement of $e$ in the paths $\pi' \cup \pi^a$.
We know that $e \notin \pi^a$, since $\pi^a$ is unblocked.
We know that $e$ is not the last edge of $\pi'$, since the last edge of $\pi'$ is $a'$ and we have been careful to consider an edge $a' \notin F_a$.
Thus the only remaining possibility is that $e$ is one of the first $j-1$ edges in $\pi'$.
That means that, for the edge $e \in \pi'$, the edge $a$ counts towards the chain-blocked property: $a$ is incident on a node ($t$) that strictly follows $e$ along $\pi'$, and we have $e \in F_a$.
Since we require that $\pi'$ is chain-unblocked, there can be at most $C \faultdeg k$ edges $a$ with this property.
However, we have shown that our analysis applies to at least $>C \faultdeg k$ different edges $a \in A_{\text{late}}$.
This completes the contradiction and the proof.
\end{proof}

\subsection{Counting Lemma}

Our next step is to prove lower bounds on the number of MUCk paths in $H$.
Our strategy is to first count the number of monotonic $k$-paths, and then separately count the number of blocked and chain-blocked $k$-paths, and argue that most monotonic $k$-paths must in fact be MUCk paths (in the relevant parameter regime).

We start with the monotone counting lemma.
Only the ``full'' monotone counting lemma in the following sequence is ultimately used; the previous two lemmas are used to help bootstrap the proof.

\begin{lemma} [Weak Monotone Counting Lemma] \label{lem:warmupwc}
If $|E(H)| \ge kn/2$, then $H$ contains a monotonic $k$-path.
\end{lemma}
\begin{proof}
This follows from a folklore theorem in graph theory sometimes called the ``hiker lemma'' (see \cite{bodwin2023alternate} for context).
We generate $n$ different monotonic paths by the following process.
Start by placing a hiker at each node of $H$.
Then, consider the edges in $H$ in increasing order.
Each time we consider an edge $(u, v)$, we ask the two hikers who currently stand at $u, v$ to walk across the edge, switching places with each other.

Once all edges have been considered, our $n$ hikers have traversed $2|E(H)| \ge kn$ edges in total.
Additionally, the path walked by each individual hiker is monotonic.
Thus, there exists a hiker who walked a monotonic path of length $\ge k$.
\end{proof}

\begin{lemma} [Medium Monotone Counting Lemma]
If $|E(H)| \ge kn$, then $H$ contains at least $kn/2$ monotonic $k$-paths.
\end{lemma}
\begin{proof}
By the weak counting lemma, we would need to delete at least $kn/2$ edges from $H$ in order to destroy all monotonic $k$-paths.
This implies that $H$ has at least $kn/2$ monotonic $k$-paths to begin with, since otherwise we could easily delete any one edge per path to destroy them all.
\end{proof}

\begin{lemma} [Full Monotone Counting Lemma] \label{lem:fullcount}
Let $d$ be the average degree in $H$.
If $d \ge 2k$, then $H$ contains at least $kn \cdot \Omega(d/k)^k$ monotonic $k$-paths.
\end{lemma}
\begin{proof}
Let $H'$ be a random edge-subgraph of $H$ on exactly $kn$ edges.
Let $x$ be the number of monotonic $k$-paths in $H$, and let $x'$ be the number of monotonic $k$-paths that survive in $H'$.
On one hand, by the medium counting lemma, we have $x' \ge kn/2$ (deterministically).
On the other hand, for any monotonic $k$-path $\pi$ in $H$, the probability that $\pi$ survives in $H'$ is:
\begin{align*}
& \left(\frac{kn}{|E(H)|} \right) \left( \frac{kn-1}{|E(H)|-1} \right) \dots \left(\frac{kn-(k-1)}{|E(H)| - (k-1)} \right)\\
&\le \left( \frac{kn}{|E(H)|} \right)^k\\
&= \Theta\left( \frac{k}{d} \right)^k.
\end{align*}
Thus we have
\begin{align*}
\frac{kn}{2} \le \mathbb{E}[x'] \le x \cdot \Theta\left( \frac{k}{d} \right)^k.
\end{align*}
Rearranging terms, we get our desired inequality of
\begin{align*}
x \ge kn \cdot \Theta\left( \frac{d}{k} \right)^k. \tag*{\qedhere}
\end{align*}
\end{proof}

Finally, we count the number of blocked or chain-blocked $k$-paths in $H$, with the plan to argue that these are much fewer than the number of monotonic $k$-paths.
The following arguments require a stronger assumption that the \emph{max} degree is bounded; in the next part we will show how this assumption can be justified.
\begin{lemma} [Chain-Blocked Path Upper Bound] \label{lem:chainub}
Let $d$ be a parameter and suppose that the maximum node degree in $H$ is $O(d)$.
Then $H$ has at most $kn \cdot O(d)^k / C$ $k$-paths that are edge-simple and chain-blocked.
\end{lemma}
\begin{proof}
We can overcount our chain-blocked paths as follows:
\begin{itemize}
\item First, choose its ending node $v_k$ ($n$ choices).

\item Then, choose an index $1 \le i \le k$ ($k$ choices) such that the $i^{th}$ edge causes $\pi$ to be chain-blocked.
That is, we commit to selecting the $i^{th}$ edge $e_i \in \pi$ in such a way that there exist $>C \faultdeg k$ edges $e$ incident to nodes in $\pi$ that strictly follow $e_i$, with $e_i \in F_e$.

\item Now build $\pi$ iteratively.
If currently we have determined a suffix $(v_j, \dots, v_k)$ of $\pi$, then we choose the previous node $v_{j-1}$ as follows.
\begin{itemize}
\item If $j-1 \ne i$, then we choose any of the $O(d)$ neighbors of $v_j$ besides $v_{j+1}$ to serve as $v_{j-1}$.

\item If $j-1=i$, then our choice is restricted, since we need this edge to cause $\pi$ to be chain-blocked.
Let us count the number of neighbors of $v_j$ that we may choose in this step.
There are $O(dk)$ edges $e$ incident to any of the previously-selected nodes $(v_{j+1}, \dots, v_k)$.
By the fault-degree constraint, each of these edges has at most $\faultdeg$ edges incident to $v_j$ in its set $F_e$.
We thus count that there are $O(d/C)$ edges incident to $v_j$ that are in $>C \faultdeg k$ such sets $F_e$.
We must choose one such edge in order for this edge to be chain-blocked, and so there are $O(d/C)$ choices of edge at this step.
\end{itemize}
\end{itemize}
Putting it together, we have $n \cdot k \cdot O(d)^{k-1} \cdot O(d/C) = kn \cdot O(d)^k / C$ edge-simple chain-blocked $k$-paths.
\end{proof}

\begin{lemma} [Blocked Path Upper Bound] \label{lem:blockub}
Let $d$ be a parameter and suppose that the maximum node degree in $H$ is $O(d)$.
Then $H$ has at most $k^2 n \cdot O(d)^{k-1} \cdot \faultdeg$ $k$-paths that are edge-simple and blocked.
\end{lemma}
\begin{proof}
Recall that a $k$-path $\pi$ is blocked if it contains edges $e, e' \in \pi$ with $e' \in F_e$.
We may overcount our blocked $k$-paths as follows:
\begin{itemize}
\item Choose distinct indices $1 \le i, i' \le k$ where $e, e'$ will occur ($O(k^2)$ choices).
In the following we will assume w.l.o.g.\ that $i < i'$; the case where $i > i'$ is symmetric, but we build $\pi$ from its last node towards the beginning instead of from its first node towards the end.

\item Choose the first node $v_1$ of $\pi$ ($n$ choices).

\item Now we build $\pi$ iteratively.
If currently we have determined a prefix $(v_1, \dots, v_j)$, then we choose the following node $v_{j+1}$ as follows.
\begin{itemize}
\item If $j+1 \ne i'$, then we choose any of the $O(d)$ neighbors of $v_j$ besides $v_{j-1}$ to serve as $v_{j+1}$.
\item If $j+1=i'$, then our choice is restricted, since we need this edge to lie in $F_e$ (since $i' > i$, note that we have already determined the edge $e$).
By the fault-degree constraint, there are at most $\faultdeg$ edges incident on $v_j$ in $F_e$, and hence there are at most $\faultdeg$ choices of edges at this step.
\end{itemize}
\end{itemize}
Putting it together, we have $O(k^2) \cdot n \cdot O(d)^{k-1} \cdot \faultdeg$ blocked $k$-paths.
\end{proof}

\subsection{Proof Wrapup}

We will first need to apply a standard ``cleaning'' step to $H$, in order to reduce to the case where the max degree and average degree are roughly equal, so that all of the above counting lemmas may be applied simultaneously.
This step is standard and has appeared in some form in a lot of prior work on spanners, but we will outline the proof anyways.

\begin{lemma} [Degree-Cleaning Lemma]
Let $H$ be a graph with $n$ nodes and a $\faultdeg$-FD $2k$-blocking set with average degree $d$.
Then there exists a graph $H'$ with $\Theta(n)$ nodes and a $\faultdeg$-FD $2k$-blocking set in which all nodes have degree $\Theta(d)$.
\end{lemma}
\begin{proof}
Let $H'$ be a graph obtained from $H$ by the following process.
Fix $d$ as the initial average degree in $H$.
While there exists a node of degree $< d/4$, delete that node and all its incident edges from $H$.
Then, while there exists a node of degree $>d$, split that node into two copies, with its edges divided equitably between the copies.
Define $H'$ as the final graph at the end of this process.

It is immediate from the construction that all surviving nodes in $H'$ have degree $\Theta(d)$; in particular this means that the average degree changes by at most a constant factor from $H$ to $H'$.
Additionally, notice that we delete at most $< nd/4 \le |E(H)|/2$ edges from $H$ to $H'$, and so the number of edges changes by at most a constant factor from $H$ to $H'$.
Since the average and total degrees both change by a constant factor, it follows that the number of nodes also changes by only a constant factor from $H$ to $H'$.
Thus $H'$ has $\Theta(n)$ nodes, as desired.

The last detail is to argue that $H'$ still has a $\faultdeg$-FD $2k$-blocking set.
This follows from the fact that we create $H'$ from $H$ only by the operations of deleting and splitting nodes.
These operations cannot create new $2k$-cycles, and they can only reduce the fault degree of the blocking set.
Thus, the remaining part of the $\faultdeg$-FD $2k$-blocking set for $H$ is still a valid $\faultdeg$-FD $2k$-blocking set for $H'$.
\end{proof}

Since our previous lemmas apply to \emph{any} graph with a $\faultdeg$-FD $2k$-blocking set, we can apply them to $H'$.
On one hand, by the dispersion lemma, $H'$ has $n^2 \cdot O(C \faultdeg k)^{k-1}$ total MUCk paths.
On the other hand, by Lemmas \ref{lem:fullcount}, \ref{lem:chainub}, and \ref{lem:blockub}, the number of MUCk paths in $H'$ is at least
$$\underbrace{\left(kn \cdot \Omega\left(\frac{d}{k}\right)^k\right)}_{\text{l.b. on monotone $k$-paths}} - \underbrace{\left(\frac{kn \cdot O(d)^k}{C}\right)}_{\text{u.b. on chain-blocked $k$-paths}} - \underbrace{k^2 n \faultdeg \cdot O(d)^{k-1}}_{\text{u.b. on blocked $k$-paths}}.$$

If we set $C \ge \Omega(k)^k$ with a large enough implicit constant, then the monotone $k$-paths dominate the chain-blocked $k$-paths.
If we have $d \ge \Omega(k)^k \cdot \faultdeg$ with a large enough implicit constant, then the monotone $k$-paths dominate the blocked $k$-paths.
So, under both these assumptions, we have
$$kn \cdot \Omega\left(\frac{d}{k}\right)^k \le \left(\text{\# MUCk paths} \right) \le n^2 \cdot O(C \faultdeg k)^{k-1}$$
and so
$$\frac{d}{k} \le n^{1/k} \cdot O(C \faultdeg k)^{(k-1)/k}$$
and so
$$|E(H)| = \Theta(nd) \le n^{1+1/k} \cdot \faultdeg^{1-1/k} \cdot k^{O(k)}$$
where the last step uses our setting of $C$.
Finally, in the case where $d \le O(k)^k \cdot \faultdeg$, then we immediately have $|E(H)| = \Theta(nd) \le n \faultdeg \cdot O(k)^k$, and so the bound still holds.
In either case, Theorem \ref{thm:extremalfdspan} is proved.

\subsection{The \textsc{Min Max Length-Bounded Cut} Problem}

Our next goal is to develop a tweak on the FD-greedy algorithm that runs in polynomial time.
Before explaining this, we will need to develop a subroutine to be used in the algorithm.
We will discuss the following problem:
\begin{mdframed}[backgroundcolor=gray!20]
\noindent \textsc{Min Max Length-Bounded Cut (LBC)}:
Given an $n$-node graph $G = (V, E)$, an integer $k$, and nodes $u, v$, find the least integer $\faultdeg$ such that there exists an edge set $F$ of degree $\deg(F) \le \faultdeg$ with $\dist_{G \setminus F}(u, v) > k$.
\end{mdframed}

For motivation, notice that for the FD-greedy algorithm in the setting where the input graph $G$ is unweighted, the test to keep or discard an edge $(u, v)$ is precisely \textsc{Min Max LBC}.
This problem is likely hard: it is a tweak on the NP-hard problem \textsc{Length-Bounded Cut (LBC)} \cite{BEHKSS06}, in which the goal is to minimize $|F|$ subject to the same constraint.
So, we will develop an approximation algorithm.
This high-level approach to speeding up the greedy algorithm was first used by Dinitz and Robelle \cite{DR20}, although the min-max degree objective introduces some new challenges that mean our approximation algorithm works quite differently from theirs.

\begin{theorem} \label{thm:approxlbc}
There is an $O(k \log n)$ approximation algorithm for \textsc{Min Max LBC} that runs in polynomial time.
Specifically, on an instance of \textsc{Min Max LBC} with solution $\faultdeg^*$, the algorithm returns a value $\widehat{\faultdeg}$ satisfying
$\faultdeg^* \le \widehat{\faultdeg} \le B \faultdeg^* \cdot k \log n,$ where $B>0$ is some absolute constant.
\end{theorem}

We will use an LP-rounding algorithm.
The following LP encodes a fractional relaxation of \textsc{Min Max LBC}:

\begin{mdframed}[hidealllines=false,backgroundcolor=gray!30]
\center \textbf{LP for Fractional \textsc{Min Max LBC}}

\flushleft
Choose $c \in \rr^{E(G)}, \faultdeg \in \rr$

To minimize $\faultdeg$

Subject to:
\begin{itemize}
\item For all nodes $v$, $\sum \limits_{e \text{ incident to } v} c_e \le \faultdeg$ \texttt{// degree constraints}

\item For all simple $u \leadsto v$ paths $\pi$ of length $|\pi| \le k$, $\sum \limits_{e \in \pi} c_e \ge 1$ \texttt{// path cut constraints}

\item For all edges $e$, $0 \le c_e \le 1$. \texttt{// cut value constraints}
\end{itemize}
\end{mdframed}

This LP could have exponentially many constraints -- roughly $n^k$ -- since it has one constraint per short $u \leadsto v$ path.
However, we next point out that it has a separation oracle, and thus can be efficiently solved.
\begin{lemma}
There is a polynomial-time separation oracle for the above LP.
\end{lemma}
\begin{proof}
Given a setting of $c, \faultdeg$, it is trivial to check the degree constraints and the path cut constraints in $\text{poly } n$ time.
To check the path constraints, let $G'$ be the weighted graph with the same edge set as $G$ and with the variables $c_e$ interpreted as edge weights.
Our goal is to check whether or not there exists a $u \leadsto v$ path $\pi$ in $G'$ that uses $|\pi| \le k$ hops, and which has total weight $w(\pi) < 1$.

To check this, we design an auxiliary graph $G''$ by creating $k+1$ independent copies $G'_0, G'_1, \dots, G'_k$ of the graph $G'$.
We interpret these as $k+1$ layers, and for a node $x$, we write $x_i$ to mean the copy of $x$ in $G'_i$.
For each edge $(x, y) \in E(G')$, for each $0 \le i \le k-1$ we add directed edges $(x_i, y_{i+1})$ and $(y_i, x_{i+1})$ to $E(G'')$, both with the same weight as $(x, y)$.
Additionally, for each node $x$ we add all directed edges of the form $(x_i, x_{i+1})$ with weight $0$.

We now claim that a given choice of $c, \faultdeg$ satisfies the path cut constraints iff $\dist_{G''}(u_0, v_k) \ge 1$, which is easy to test in polynomial time.
Indeed, notice that any $u_0 \leadsto v_k$ path $\pi$ in $G''$ corresponds to a $u \leadsto v$ path in $G$ on $\le k$ hops, and that its length $w(\pi)$ is the sum of cut variables $c_e$ along $\pi$.
\end{proof}

Now let $c, \faultdeg_{LP}$ be an optimal solution to this LP.
Notice that the IP in which each $c_e \in \{0, 1\}$ encodes \textsc{Min Max LBC} exactly.
Hence, letting $\faultdeg^*$ be the true solution to the (non-relaxed) \textsc{LBC} problem, we have
$$\faultdeg_{LP} \le \faultdeg^*.$$

We will next use the entries of $c$ to randomly generate an edge set $F$ (this step is equivalent to rounding the entries of $c$ to $\{0, 1\}$).
Specifically: each edge $e$ is included in $F$ independently with probability $\min\{A \cdot c_e \cdot k \log n, 1\}$, where $A>0$ is a large enough absolute constant that will control the probability of correctness.
In the following analysis, for each edge $e$ we will let $X_e$ be the Bernoulli random variable indicating whether or not $e \in F$.

We can bound the degree of $F$ as follows:
\begin{lemma}
With high probability, $\deg(F) \le O\left( \faultdeg^* \cdot k \log n\right)$.
\end{lemma}
\begin{proof}
Fix a node $v$.
By the degree constraints, we have
$$\sum \limits_{e \text{ incident to } v} c_e \le \faultdeg_{LP} \le \faultdeg^*.$$
It follows that
$$\mathbb{E}[\deg_F(v)] \le \sum \limits_{e \text{ incident to } v} A \cdot c_e \cdot k \log n \le A \faultdeg^* \cdot k \log n.$$
By the Chernoff bounds, we then have
$$\Pr\left[ \deg_F(v) \ge 2 A \faultdeg^* \cdot k \log n\right] \le \exp(-\Theta(A \faultdeg^* k \log n)) \le \frac{1}{n^{\Theta(A \faultdeg^* k)}}.$$
By a union bound over the $n$ nodes in the graph, we thus have
$$\Pr\left[\deg(F) \le O_A\left(\faultdeg^* \cdot k \log n\right)\right] \ge 1 - \frac{1}{n^{\Theta(A \faultdeg^* k)-1}}.$$
By choice of large enough constant $AA$, this means $\deg(F) \le O_A(\faultdeg^* \cdot k \log n)$ with high probability, completing the proof.
\end{proof}

Then we prove correctness:
\begin{lemma}
With high probability, $\dist_{G \setminus F}(u, v) > k$.
\end{lemma}
\begin{proof}
Let $\pi$ be a $u \leadsto v$ path in $G$ with $|\pi| \le k$ edges.
We will measure the probability that $F$ contains at least one edge in $\pi$.
First, by the path cut constraints, we have
$$\sum \limits_{e \in \pi} c_e \ge 1.$$
If any edge $e \in \pi$ has fractional cut value $c_e$ satisfying $A \cdot c_e \cdot k \log n \ge 1$, then this edge is included in $F$ deterministically.
Otherwise, we have
$$\mathbb{E}\left[\sum \limits_{e \in \pi} X_e\right] = \sum \limits_{e \in \pi} A \cdot c_e \cdot k \log n \ge A \cdot k \log n.$$
By Hoeffding's Inequality, we can bound
$$\Pr\left[ F \text{ does not contain an edge in } \pi \right] = \Pr\left[\sum \limits_{e \in \pi} X_e < 1\right] \le \exp\left(- \Theta\left(\frac{A^2 k^2 \log^2 n}{k}\right)\right) \le \frac{1}{n^{Ak}}$$
where the last step follows by choice of large enough constant $A$.
There are at most $n^k$ possible $u \leadsto v$ paths in $G$ with $\le k$ edges.
By a union bound, the probability that there is \emph{any} such path that does not intersect $F$ is at most $\frac{1}{n^{(A-1)k}}$.
So, with high probability, no $u \leadsto v$ path on $\le k$ edges survives in $G \setminus F$.
\end{proof}

\subsection{Improvement to Polynomial Time}

We now describe how to tweak the FD-greedy algorithm to run in polynomial time, at mild cost in spanner size.
This part very directly follows the approach of Dinitz and Robelle \cite{DR20}.
We simply run the previous FD-greedy algorithm, but we use our approximation algorithm for \textsc{Min Max LBC} to test edges:

\begin{mdframed}[hidealllines=false,backgroundcolor=gray!30]
\center \textbf{Algorithm Approximate $\GreedyFDSpanner$}
\begin{flushleft}
Input: A graph $G$ and a positive integer stretch parameter $k$. \\
Output: An $\faultdeg$-FD $(2k-1)$-spanner $H \subseteq G$ with at most $\faultdeg^{1-1/k} \cdot n^{1+1/k} \cdot O(k)^k$ edges.
\end{flushleft}
\begin{itemize}
\item Let $H \gets (V, \emptyset, w)$ be the initially-empty spanner
\item For each edge $(u, v) \in E$ in order of nondecreasing $w(u, v)$

\begin{itemize}
\item If \textsc{ApproxMinMaxLBC}$(u, v, 2k-1) \le B \faultdeg k \log n$ \texttt{// from Theorem \ref{thm:approxlbc}} 
\begin{itemize}
\item Add $(u, v)$ to $H$
\end{itemize}
\end{itemize}

\item Return $H$
\end{itemize}
\end{mdframed}

This algorithm now runs in polynomial time, but due to the changes we need to re-establish its correctness and size bound.
This mostly follows by rehashing our early arguments used for the (non-approximate) $\GreedyFDSpanner$.
\begin{lemma}
The output spanner $H$ of the Approximate $\GreedyFDSpanner$ algorithm is a $\faultdeg$-FD $(2k-1)$-spanner of $G$.
\end{lemma}
\begin{proof}
This is nearly the same argument as Lemma \ref{lem:greedyspannercorrect}, with a couple minor tweaks.
As before, let $F$ be an arbitrary set of faulty edges with $\deg(F) \le \faultdeg$, and it suffices to verify the spanner stretch inequality for an arbitrary edge $(u, v) \in G \setminus F$.
If $(u, v)$ is added to the spanner $H$ by the greedy algorithm, then we immediately have $\dist_H(u, v) = w(u, v)$.
Otherwise, if $(u, v)$ is discarded from $H$ by the algorithm, then it must be that $\textsc{ApproxMinMaxLBC}(u, v, 2k-1) > B \faultdeg k \log n$.
By Theorem \ref{thm:approxlbc}, this means that at the time $(u, v)$ is considered, there is no fault set of degree $\deg(F) \le \faultdeg$ that intersects all $u \leadsto v$ paths on $\le 2k-1$ edges.
Hence at least one such path $\pi$ survives in $G \setminus F$.
Since the greedy algorithm considers edges in increasing order of weight, every edge in $\pi$ has weight $\le w(u, v)$.
We therefore have
$$\dist_{G \setminus F}(u, v) \le |\pi| \cdot w(u, v) \le (2k-1) \cdot w(u, v)$$
and so the spanner inequality holds.
\end{proof}

\begin{lemma}
The output spanner $H$ of the Approximate $\GreedyFDSpanner$ algorithm has size
$$|E(H)| \le \Oish\left(\faultdeg^{1-1/k} n^{1+1/k}\right) \cdot O(k)^k.$$
\end{lemma}
\begin{proof}
From Theorem \ref{thm:approxlbc}, the algorithm $\textsc{ApproxMinMaxLBC}$ overestimates the true value of $\textsc{MinMaxLBC}$ on a given input instance.
Thus, each time an edge $(u, v)$ is added to the spanner $H$, there exists a length-bounded cut $F$ of degree $\deg(F) \le \Oish_B(\faultdeg k)$.
Exactly as in Lemma \ref{lem:blockingsetexists}, the output spanner $H$ therefore has a $\Oish(\faultdeg k)$-FD blocking set, and the edge bound then follows by plugging this into Lemma \ref{lem:blockingsetsizebound}.
\end{proof}

\section{Existential Lower Bounds for FD Spanners and Connectivity Certificates}\label{sec:LB}

\subsection{FD Spanner Lower Bounds \label{sec:exspanlb}}

We first show that the upper bound in Theorem \ref{thm:extremalfdspan} is optimal, up to its dependence on $k$:
\begin{theorem}
Assuming the girth conjecture \cite{erdHos1964extremal}, for all $k, \faultdeg$, there are $n$-node undirected unweighted graphs in which every $\faultdeg$-FD $(2k-1)$-spanner has at least $\Omega(\faultdeg^{1-1/k} n^{1+1/k})$ edges.
\end{theorem}

The proof of this theorem is essentially from \cite{BDPW18}, where the following construction is used for a lower bound against vertex fault tolerant spanners.

\paragraph{The Construction.}

We will analyze the following graphs:

\begin{itemize}
\item The girth conjecture \cite{erdHos1964extremal} states that for all positive integers $k$, there exist $n$-node graphs on $\Theta(n^{1+1/k})$ edges with girth $>2k$.
We begin by letting $\Gamma = (V, E)$ be one such graph.

\item Replace each node $v \in V$ with $\faultdeg$ copies, $v_1, \dots, v_{\faultdeg}$.
Each edge $(u, v) \in V$ is replaced by $\faultdeg^2$ edges, arranged in a biclique on the vertex sets $\{u_1, \dots, u_{\faultdeg}\} \times \{v_1, \dots, v_{\faultdeg}\}$.
Let $G$ be the final graph after these replacements.
\end{itemize}

The graph $G$ has $N := n \faultdeg$ nodes, and its number of edges is
$$|E(G)| = \Theta\left( \faultdeg^2 \cdot n^{1+1/k} \right) = \Theta\left( \faultdeg^{1-1/k} \cdot N^{1+1/k} \right).$$
Thus, it suffices to argue that the only $\faultdeg$-FD $(2k-1)$ spanner of $G$ is $G$ itself.

\paragraph{Analysis.}
Let $H$ be a $\faultdeg$-FD $(2k-1)$ spanner of $G$, and let $(u_i, v_j) \in E(G)$ (so our goal is to argue that $(u_i, v_j) \in E(H)$ as well).
Let $F$ be the fault set that contains every edge connecting any copy of $u$ to any copy of $v$, except for the edge $(u_i, v_j)$.
That is:
$$F := \{(u_a, v_b) \ \mid \ a \ne i \text{ or } b \ne j\}.$$
Notice that the degree of $F$ on $u_i, v_j$ is $\faultdeg-1$, and for each copy of $u, v$ besides $u_a, v_b$, the degree of $F$ on this node is $\faultdeg$.
Thus $F$ is a $\faultdeg$-FD fault set, which means that by definition of $\faultdeg$-FD spanners, we require
$$\dist_{H \setminus F}(u_i, v_j)\le (2k-1) \cdot \dist_{G \setminus F}(u_i, v_j) = 2k-1.$$
This inequality clearly holds if $(u_i, v_j) \in E(H)$ (in which case we have $\dist_{H \setminus F}(u_i, v_j)=1$).
On the other hand, suppose for contradiction that $(u_i, v_j) \notin E(H)$.
In particular this means there is \emph{no} edge connecting any copy of $u$ to any copy of $v$ remaining in $H \setminus F$.
Thus, any $u_i \leadsto v_j$ path in $H \setminus F$ corresponds to a $u \leadsto v$ path $\pi$ in $\Gamma$ that does not include the edge $(u, v)$.
Since $\Gamma$ has girth $>2k$, and $\pi \cup \{(u, v)\}$ forms a cycle in $\Gamma$, there must be $>2k-1$ edges in $\pi$.
So any $u_i \leadsto v_j$ path in $H \setminus F$ has length $>2k-1$, violating the above spanner inequality, and completing the contradiction.

We have now proved that every edge $(u_i, v_j) \in E(G)$ must lie in $E(H)$ as well, which means that $G=H$, i.e., $G$ is the only $\faultdeg$-FD $(2k-1)$-spanner of itself, which completes the proof.

\subsection{FD Connectivity Certificate Lower Bounds}

The naive lower bound for a $\faultdeg$-FD connectivity certificate is $\Omega(n\faultdeg)$ edges, from any $\faultdeg$-regular graph.
Here we point out that there is a minor improvement to this lower bound available:

\begin{theorem}
For all $\faultdeg$, there are $n$-node undirected unweighted graphs in which every $\faultdeg$-FD connectivity certificate has $\Omega(\faultdeg n \cdot \frac{\log n}{\log \faultdeg})$ edges.
\end{theorem}

\paragraph{The Construction.}

Let $d$ be a parameter of the construction.
Our lower bound graph $G$ is defined as follows:
\begin{itemize}
\item The nodes of $G$ are $[\faultdeg]^d$, i.e., $d$-tuples of integers between $1$ and $\faultdeg$.

\item Two nodes of $G$ are joined by an edge iff their $d$-tuples are equal on all but one coordinate.
\end{itemize}

This graph has $n := \faultdeg^d$ nodes, and each node has degree
$$(\faultdeg-1)d = (\faultdeg-1) \cdot \frac{\log n}{\log \faultdeg}.$$
Thus it suffices to argue that $G$ is the only $\faultdeg$-FD connectivity certificate of itself.

\paragraph{The Analysis.}
Let $(u, v) \in E(G)$, where the $d$-tuples associated to nodes $u, v$ agree on all but the $i^{th}$ coordinate.
Let $H$ be a $\faultdeg$-FD connectivity certificate of $G$, and suppose for contradiction that $(u, v) \notin E(H)$.
Let $F$ be the fault set that contains every edge in $G$ that connects two nodes whose $d$-tuples agree on all but the $i^{th}$ coordinate, except that we also omit $(u, v)$ itself from $F$.
Notice that the degree of $F$ on any node $x$ is at most $\faultdeg-1$, since there are only $\faultdeg-1$ other nodes whose $d$-tuples agree with $x$ on all but the $i^{th}$ coordinate.

We have that $u, v$ are connected in $G \setminus F$, since $(u, v) \in E(G) \setminus F$.
However, we claim that $u, v$ are disconnected in $H \setminus F$.
This follows by noticing that all nodes reachable from $u$ in $H \setminus F$ must have the same $i^{th}$ coordinate as $u$, but $v$ has a different $i^{th}$ coordinate from $u$.
This completes the contradiction.

We have shown that, for all edges $(u, v) \in E(G)$, we have $(u, v) \in E(H)$.
Thus $H=G$, and so $G$ is the only $\faultdeg$-FD connectivity certificate of itself, and so the theorem follows.



\bibliographystyle{alpha}
\bibliography{refs}

\appendix

\section{Missing Proofs}\label{sec:missing}

\APPENDDAG

\APPENDPOLYROUTE
\APPENDLEMEMBEXP

\section{$G(n,p)$ is a good constant-hop router}

\begin{lemma}\label{lem:gnprouting}
For any $n$, $p = \omega(\frac{\log n}{n})$ let $G$ be a randomly sampled $G(n,p)$ graph. With high probability, any unit demand $D$ in $G$ can be routed with dilation $\dilation = O(\log_{np} n)$ and congestion $\congestion = O(\log n)$.
\end{lemma}
\begin{proof}
Define $\eps > 0$ such that $np=n^{\eps}$. It is well known that for any sufficiently large $t=\Theta(\frac{1}{\eps})$ a $t$-step random walk in $G(n,p)$ from any starting vertex suffices to polynomially close to its stationary distribution. This means no matter what node one starts from a $t$-step random walk will, up to polynomially small deviations, end at every node with probability proportional to its degree. Degrees in $G(n,p)$ with $p = \omega(\frac{\log n}{n})$ and average degree $np = \omega(\log n)$ are furthermore highly concentrated.

We will show that these facts imply that, with high probability, any unit demand $D$ in $G$ can be routed with dilation at most $2t$ and congestion at most $O(\log n)$ as claimed. 

In particular, suppose each node $v$ starts with $deg(v) \cdot 2m$ many pebbles and each pebble (simultaneously) performs a random walk for $t$ steps. Let $P$ be the routing (i.e., collection of paths) traced out by these pebbles. Because of the above mixing property it is the case that, with high probability and up to $(1 + o(1))$-factors, this routing $P$ ends up again with exactly $deg(v) \cdot 2m = deg(v) \cdot sum_u deg(u)$ many pebbles at node $v$ of which $deg(u)$ many started from node $u$. This means $P$ integrally routes exactly the demand $m \cdot D_{mix}$ where $D_{mix}(u,v) = \frac{deg(v) \cdot deg(u)}{m}$ is the so-called uniform mixing demand. The dilation of this routing is at most $t$ and the expected congestion of the routing is at most $mt$. We can obliviously route any unit demand $D$ in $G$ by using $P$ to mix and unmix. In particular, for each $u,v$ we construct $D(u,v)$ many routing paths by independently sampling an intermediate node $x$ in $G$ (sampling nodes proportional to their degree) and then independently sampling a path from $u$ to $x$ and from $x$ to $v$ from $P$ and concatenating these two paths into a routing path from $u$ to $v$. The dilation of this routing is at most $2t$ and the expected congestion for this routing is at most $2t$ on any edge $e$. After making routing paths simple the multiplicative Chernoff bound guarantees that any edge has congestion at most $O(\log n)$ with high probability. 
\end{proof}

\end{document}